\documentclass[sigconf]{acmart}

\usepackage{booktabs}
\setcopyright{none}

\usepackage{color,ifthen,graphicx,enumerate,float,url,cases}
\usepackage{mathrsfs}
\usepackage{epstopdf}
\usepackage{nicefrac}
\usepackage{pstool}
\usepackage{subfigure}
\usepackage{caption}
\usepackage{amsmath,amssymb}
\usepackage{bm}
\usepackage{breakurl}
\usepackage{caption}
\usepackage{latexsym}
\usepackage{amsmath}
\usepackage{amsfonts}
\usepackage{amsthm}
\usepackage{extarrows}
\usepackage{multirow}
\usepackage{epsfig}
\usepackage{indentfirst}
\usepackage{multirow}
\usepackage{algorithm}
\usepackage{algorithmic}
\usepackage{float}
\newtheorem{theorem}{Theorem\,}

\newtheorem{proposition}[theorem]{Proposition\,}

\newcommand{\var}{\mathop{\text{var}}}

\newcommand{\dif}{\mathrm{d}}
\newcommand{\PoI}{\mathop{\text{PoI}}}
\newcommand{\cov}{cov}
\begin{document}
\settopmatter{printacmref=false} % Removes citation information below abstract
\renewcommand\footnotetextcopyrightpermission[1]{} % removes footnote with conference information in first column
\pagestyle{plain} % removes running headers

\title{Optimal pricing
for peer-to-peer sharing with network externalities }
\titlenote{Part of the results will appear in the Proceedings of NetEcon'18.}
%\author{}
%\affiliation{%
  %\institution{}
%  \streetaddress{Address}
%  \city{City}
%  \state{State}
%  \postcode{Zipcode}
%}
%\email{}
\author{Yunpeng~Li}
\affiliation{\institution{Singapore University of Technology and Design}}
\email{yunpeng\_li@mymail.sutd.edu.sg}
\author{Costas~Courcoubetis}
\affiliation{\institution{Singapore University of Technology and Design}}
\email{costas@sutd.edu.sg}
\author{Lingjie~Duan}
\affiliation{\institution{Singapore University of Technology and Design}}
\email{ lingjie\_duan@sutd.edu.sg}
\author{Richard~Weber}
\affiliation{\institution{University of Cambridge}}
\email{ rrw1@cam.ac.uk}
%\author{Yunpeng~Li,~Costas~Courcoubetis,~and~Lingjie~Duan
%\thanks{Y. Li, C. Courcoubetis and L. Duan are with the Engineering Systems and Design Pillar, Singapore University of Technology and Design, Singapore 487372,
%Singapore (e-mail: yunpeng\_li@mymail.sutd.edu.sg, \{costas, lingjie\_duan\}@sutd.edu.sg).
%}}

\begin{abstract}
   In this paper, we analyse how a peer-to-peer sharing platform should
  price its service (when imagined as an excludable public good) to
  maximize profit, when each user's participation adds value to the
  platform service by creating a positive externality to other
  participants.
  To characterize  network externalities as a function of the number of participants, we consider different bounded and unbounded user utility models. The bounded utility model fits many infrastructure sharing applications with bounded network value,  in which complete coverage has a finite user valuation (e.g., WiFi or hotspot).
  The unbounded utility model fits the large scale data sharing and explosion  in social media, where it is expected that the network value follows Metcalfe's or Zipf's  law.
  For both models, we analyze the optimal pricing schemes to select heterogeneous users in the platform under complete and incomplete information of users' service valuations. We propose the concept of price of information (PoI) to characterize the profit loss due to lack of information, and present provable PoI bounds for different utility models.
  We show that the $\PoI=2$ for the bounded utility
  model, meaning that just half of profit is lost, whereas
the $\PoI\geq 2$ for the unbounded utility model and increases as
  for a less concave utility function.
We also show that the
  complicated differentiated pricing scheme which is optimal under
  incomplete user information, can be replaced by a single uniform
  price scheme that is asymptotic optimal.
  Finally, we extend our pricing schemes  to a  two-sided market by including a new group of `pure' service users contributing no externalities, and show that the platform may charge zero price to the original group of users in order to attract the pure user group.
\end{abstract}

\maketitle
\section{Introduction}
%First paragraph: We should first say with more powerful terminal devices (e.g., smartphones), many peer-to-peer platforms emerge ... where individuals can not only be users but also contributors to the platform services... Then give two typical examples: FON (for relating to the first model) and also online games/social networks... (second model)

%Second paragraph:  say such peer-to-peer sharing platfroms creates ...money...many subscribers (cite some news or white papers to tell huge revenue growth... billions). Say, how to price platform services is the key question for such platforms ... Yet the key challenge for such pricing design is the incomplete information about users... and there is a bad need of a simple pricing scheme easy to implement...

%Third paragraph should first introduce the related work. We should first introduce technical works for pricing mechanism design... truthful... user types... Yet they are complicated to implement... (say something critical). Then introduce more applied works regarding peer-to-peer sharing platforms...Elaborate their results... We should say they do not consider platform participants as both users and contributors... Technically, different from traditional pricing mechanisms... we want to design simple...

Due to advances in wireless technology and more powerful mobile
devices (e.g., smartphones), it is common today that when users join a
peer-to-peer sharing platform they not only enjoy the provided service
but also contribute to the service's value. There are roughly two types of peer-to-peer sharing platforms: infrastructure and content sharing \cite{8327515}. The former type of platforms allows users to cooperate and contribute physical resources to create networking or computing services. For example, FON is a WiFi sharing platform  whose user opens its home WiFi connection to the community and can access the others' WiFi access points  \cite{4509810}. The latter type of platforms includes online social media (e.g., WeChat, WhatsApp), where platform users create and share massive content with each other and their number has reached 1.6 billion in 2014. The global revenue of such peer-to-peer sharing platforms is fast growing and is expected to increase to US\$40 billions by 2022  \cite{Juniper}. How to price their services for selected users under network externalities is a key question for such profit-maximizing platforms.

Peer-to-peer sharing economy of such excludable public goods has been
widely studied in the recent literature.  \cite{ANTONIADIS2004133} and
\cite{10.1007/3-540-45598-1_9} study how to address the incentive
issues for efficient sharing in peer-to-peer networks via mechanism
design.  Courcourbetis and Weber in \cite{1626429} study pricing of an
infrastructure-sharing platform (e.g., peer-to-peer file sharing) and
find the network value (profit) in an asymptotic sense and find that
network value/profit is bounded when each user randomly caches and
shares a subset of distinct files.
Metcalfe and Zipf's laws study the network value for the social media
platforms, showing the service value to an individual increases
super-linearly with the total user number and is thus unbounded
\cite{Metcalfe}.  In \cite{4509810}, \cite{8327515} and
\cite{8016389}, users' dual modes (i.e., contributors and consumers)
are considered and optimal pricing schemes for network externalities
is designed under complete information. Assuming full information of
users' private utilities, \cite{doi:10.1287/opre.1120.1066}
investigates the optimal pricing according to the network structure,
and proposes a simplified approximation  using uniform pricing, i.e.\ every
 users sees the same price.
 \cite{7835123} further consider that the network externalities can be positive or negative, affecting the final pricing design.
%, the authors design the optimal pricing strategies for a service provider in a social network. Different from their work,  we consider a general public good model and consider pricing strategies given partial information about users' service valuations. \cite{musacchio2006wifi,4509810,6566899,7562462}  study the pricing problem for provisioning wireless network services.
Different from these works, we consider the challenging scenario of incomplete information for optimal pricing design of excludable public goods, and study the feasibility to employ a simple pricing approach for profit maximization (without users' reporting of private information as in VCG auction). The newly proposed concept, price of information is unique to characterize the profit loss due to lack of information.

 %Related work on two-sided market can be found in \cite{rochet2003platform,rochet2006two}. Different from classic two-sided market, in our model, we consider platform participants as both users and contributors.
%In this paper, we address the pricing  problem of  a for-profit platform. Given complete information about users' valuation, the platform can use differentiated pricing to transform the consumer surplus into profits. Given  incomplete information about users' valuation, the platform can still use differentiated pricing but there is a huge profit loss since incentives are needed for users to disclose their true valuation. We compare the maximal profits obtained by platform in these two scenarios and show that the profit ratio varies depending on which utility model we use. We also consider a uniform pricing scheme where the platform only needs to announce a single price and the users can determine whether to join the provisioning of services. We show that this uniform pricing scheme works asymptotically well compared with the differentiated pricing scheme. We show that they have equal asymptotical performances in terms of expected profits..
%%%%%%%%%%%%%%%%%%%%%%%%%%%%%%%%%%%
Our main contributions and key novelty are summarised as follows.
\begin{itemize}
\item We study the optimal pricing for a peer-to-peer sharing platform under incomplete information, by considering both the infrastructure and content sharing applications (with bounded and unbounded network externalities, see Section \ref{systemmodel}). The platform is profit-maximizing and designs pricing to include target users to contribute to the excludable public goods.
\item
For both bounded and unbounded user utility models, we analyze the optimal pricing schemes to select heterogeneous users in the platform under complete and incomplete information of users' service valuations. We propose the concept of price of information, which is defined as the ratio of profits under complete and incomplete information, to characterize the profit loss due to lack of information, and present provable PoI bounds for different utility models. We prove that the $\PoI=2$ for the bounded utility
  model, meaning that just half of profit is lost. For a general unbounded utility model, we prove the PoI is in the interval $[2,27/8]$, that is, PoI is
at least 2 and is greater for a less concave utility function.
\item
We simplify the complicated differentiated pricing scheme under incomplete information, by replacing it by a single uniform price. The uniform price mechanism does not need users to report their private information of service valuations and achieves asymptotical optimality as user number goes to infinity for both bounded and unbounded user utility models.
\item
We extend our pricing schemes  to a  two-sided market
by including a new group of `pure' service users contributing no externalities.
 We show that the platform needs to decide different pricing to different groups of users
and may charge zero price to the original group of users in order to attract the pure user group. We prove that the uniform pricing scheme is still asymptotically optimal as user number goes to infinity and that PoI increases as the fraction of original group of users decreases.
\end{itemize}

%The rest of the paper is organized as follows. Section~\ref{systemmodel} introduces the sharing platform model and discusses two utility models . Section \ref{BUF} analyses the platform's  pricing strategies given bounded utility model and Section \ref{UUF} analyses the platform's  pricing strategies given unbounded utility model. Section \ref{tsm} discusses a two-sided market in peer-to-peer service sharing network. Section \ref{simulation} presents the simulation results and  Section \ref{conclusion} concludes the paper.

\section{System Model}\label{systemmodel}

We consider a peer-to-peer platform who wants to maximize its profit. It faces a set of potential users
${N}=\{1,\dotsc,n\}$ who choose to participate in the subscribing to the platform service or not.
 Define binary variable  $\pi_i=1$ or $0$ , telling that user i will or will not participate. The vector
 $\pi=(\pi_1,\dotsc,\pi_n)$ summarizes all users' participation decisions. The total service value is denoted by $ \phi(\pi)$, which  is a function of $\pi$ to tell the network externalities.  Consider that each user contributes equally to the service as a public good, then $\phi(\pi)$ can be rewritten as a function of the number of platform users denoted by $m=\sum_{i=1}^n \pi$, that is, $\phi(\pi)=\phi(m)$. We will introduce the detailed formulation of bounded and unbounded $\phi(\cdot)$ in Sections \ref{butility} and \ref{unutility}, respectively.

 Users have heterogeneous service valuations towards the platform service. Let $\theta_i$ be the user $i$'s service valuation  and this is his private information. Without loss of generality, we assume
$\theta_1>\theta_2>\cdots>\theta_n$ and denote valuation vector  ${\theta}=(\theta_1,
\theta_2, \dotsc, \theta_n)$.
 The utility of a participant $i$ is proportional to his valuation and the total service value, that is, $\theta_i \phi(\pi)$.
The platform can charge differently for different users' subscriptions.  Let $p_i$ be the membership fee charged to  $i$.
The payoff of  user $i$  is his utility of the total service value minus the membership fee, that is,
\begin{equation}\label{up}
u_i=\pi_i(\theta_i \phi(\pi)- p_i).
\end{equation}

The platform's goal is to maximize its total profit and it may not include all users. Let $c$ be the platform cost (e.g., equipment fee for installing an access point in WiFi sharing )
  for adding a user to access shared service with the exsiting
others.
The total profit, denoted by $\Pi$, is a function of $\pi$ and $c$ as follows
\begin{equation}
\Pi=\sum_{i\in N}\pi_i(p_i-c).
\end{equation}
\subsection{Bounded User Utility Model}\label{butility}
In an infrastructure sharing platform, the service  coverage or value is bounded (e.g., by 100\% citywide), no matter how many users participate. Thus, user utility function is bounded in this model.  For modelling bounded $\phi(m)$, take WiFi sharing in a finite region of  a normalized unit square surface for example. $n$ users are randomly distributed in the square and each user can cover a circle of radius $r$ ($0<r<<1$) or an area $\pi r^2$.
The total coverage depends on the  total user number $m$. For an arbitrary point in the square surface, the probability that it is not covered by a single user is $\rho=1-\pi r^2$ and the probability that it is not covered by the $m$ users is $\rho^m$.
That is,
 \[\phi(m)=1-\rho^m,\]
which is bounded by 1 and is concavely increasing in $m$.
We  can rewrite user $i$'s payoff \eqref{up} as follows,
\begin{equation}\label{upb}
u_i=\pi_i(\theta_i(1-\rho^{\sum_{j\in N}\pi_j})- p_i).
\end{equation}
In Section \ref{BUF}, we will focus on this bounded utility model  and analyse the optimal  pricing schemes under complete and incomplete information.
\subsection{Unbounded User Utility Model}\label{unutility}
In an online social media, user utility increases super-linearly with the number of users, following Metcalfe's or Zipf's laws. Metcalfe's law suggests that a user will get equal benefits from the other $m-1$ participants. The user's utility is proportional to $m$ and when $m$ is sufficiently large,  $\phi(m)\approx m$ \cite{Metcalfe}.  Zipf's law suggests that a user will benefit from the others differently,  in inverse proportion to the frequency  with which he interacts with(i.e., frequency $1/i$ with the $i$-th closest user among $m$ users). Then  $\phi(m)=\sum_{i=1}^{m-1}(1/m)\approx \log m$  \cite{Metcalfe}.
 As a result, user $i$'s payoff \eqref{up} becomes,
 \begin{numcases}
{ u_i=}
\pi_i(\theta_i\log (\sum_{j}\pi_j)- p_i)\,,&\mbox{if Zipf's law;}\label{uplog}\\
\pi_i(\theta_i(\sum_{j}\pi_j)- p_i)\,,& \mbox{if Metcalfe's law.}\label{upu}
 \end{numcases}

\section{ Optimal Pricing for Bounded Utility Model}\label{BUF}
In this section, we will analyse  the platform's pricing strategy for bounded user utility  $\phi(m)=1-\rho^m$.
\subsection{Pricing under Complete Information }\label{combounded}
Under complete information about all users' valuations $\theta_i$'s, the platform's
optimization problem is to choose prices $p_i$'s and control admission $\pi_i$'s to maximize its
profit. The payoff of a participant in \eqref{upb} cannot be negative, otherwise he will choose not to participate.
Formally, the problem is
\begin{align}
&\max_{\{(\pi_i,\ p_i), {\ i\in{N}}\}}\ \sum_{i\in{N}}\pi_i (p_i-c)\nonumber\\
&\text{s.\ t.\ } \pi_i\bigg(\theta_i(1-\rho^{\sum_{j}\pi_j})- p_i
\bigg)\geq 0,\quad\text{for all } i\in{N}.\label{opt:completeinfob}
\end{align}
%where $c\in(0,1)$. We denote by $\Pi$ the platform's maximal profit.
At optimality, the constraints in problem
(\ref{opt:completeinfob}) are tight. For any user with $\pi_i=1$ or
$0$, it is optimal to leave a zero payoff to him
by setting the price to be
\begin{equation}%\label{eq:price_selection}
p_i^*(\pi)=\theta_i(1-\rho^{\sum_{j}\pi_j}) .\nonumber
\end{equation}
This result helps simplify problem (\ref{opt:completeinfob}) to
\begin{equation}\label{opt:completeinfo_pib}
  \max_{\{\pi_i,\ {i\in{N}}\}}\ \sum_{i\in{N}}\pi_i \bigg(\theta_i(1-\rho^{\sum_{j}\pi_j})-c\bigg).
\end{equation}
To help solve this problem, we start with a lemma about the platform's preference among users.

\begin{lemma}\label{lem:ij}
  At the optimality of problem \eqref{opt:completeinfo_pib}, for any two users $i, j\in N$ with
  $\theta_i>\theta_j$, if user $j$ is included in the platform  (i.e., $\pi_j=1$), then user $i$ should also be included
  ($\pi_i=1$).
\end{lemma}

It follows from Lemma~\ref{lem:ij} that the platform will
select $m$ users with the largest service valuations and problem \eqref{opt:completeinfo_pib}
reduces to
\begin{equation}
  \max_{m\in N}\ \bigg((1-\rho^m)\sum_{i=1}^m \theta_i - mc\bigg).
\end{equation}

This problem's objective  function is not a monotonic function of $m$ and it is not possible to derive closed-form solution of $m$. Yet we can use the efficient one-dimensional  search method to find the optimal $m$ numerically.

\subsection{Pricings under Incomplete Information }\label{incombounded}
Under incomplete information, the platform does
not know $\theta_i$'s exactly but their distributions. We assume
$\theta_i$'s are independent and identically distributed on $[0,1]$ with cumulative distribution function
$F$. The cost is comparable and we have $c\in(0,1)$.
We will derive a optimal (differentiated) pricing scheme and then propose  a uniform pricing scheme as approximation. We will compare these two different pricing schemes asymptotically.
\subsubsection{Optimal/Differentiated Pricing Scheme}\label{ODPS}
 Under incomplete information, the platform will require each user $i$ to declare his $\theta_i$. Given the $\theta_i$'s (may or may not be truthful) declared by the users, the platform should choose $p_i$'s and $\pi_i$'s as functions of the $\theta_i$'s distribution
  to maximize its profit, i.e.,
\begin{align}
&\max_{\pi_i(\cdot),\ p_i(\cdot)} E_\theta \biggl(\sum_{i=1}^n \pi_i(\theta) (p_i(\theta)-c) \biggr) \label{eq:opt_Qb}\\
&\text{subject to}\nonumber\\
& E_{\theta_{-i}}\bigg(\pi_i(\theta_i,\theta_{-i}) \bigg(\theta_i (1-\rho^{\sum_{j}\pi_j(\theta_i,\theta_{-i})})-p_i(\theta_i,\theta_{-i})\bigg) \bigg)\geq 0,\label{constraint1b}\\
&E_{\theta_{-i}}\bigg(\pi_i(\theta_i,\theta_{-i})
\bigg(\theta_i  (1-\rho^{\sum_{j}\pi_j(\theta_i,\theta_{-i})})-p_i(\theta_i,\theta_{-i})\bigg)
\bigg)\nonumber \\ & \geq E_{\theta_{-i}}\bigg(\pi_i(\theta_i',\theta_{-i}) \bigg(\theta_i  (1-\rho^{\sum_{j}\pi_j(\theta_i,\theta_{-i})})-p_i(\theta_i',\theta_{-i})\bigg) \bigg),\label{constraint2b}\\
 &\text{for\ all}\ i\ \text{and} \ \theta_i',\nonumber
\end{align}
where $\theta_{-i}=(\theta_1,\cdots,\theta_{i-1},\theta_{i+1},\cdots,\theta_n)$ is a vector consists of all the users' valuations except $\theta_i$.
Constraint \eqref{constraint1b} is to ensure  individual rationality or participation, i.e., user $i$'s expected payoff conditional on $\theta_{-i}$ is nonnegative, and constraint \eqref{constraint2b}
 is to ensure  incentive compatibility, i.e., user $i$ must declare his valuation truthfully.

 Let us define three functions:
\begin{align}
g(\theta_i)&=\theta_i - \frac{1-F(\theta_i)}{f(\theta_i)},\label{g}\\
V_i(\theta_i)&=\int \pi_i(\theta_i,\theta_{-i})  (1-\rho^{\sum_{j}\pi_j(\theta_i,\theta_{-i})})\dif F^{n-1} (\theta_{-i}),\label{V}\\
P_i(\theta_i)&=\int \pi_i(\theta_i,\theta_{-i}) p_i(\theta_i,\theta_{-i}) \dif F^{n-1}(\theta_{-i}).\label{P}
\end{align}
Note that $\theta_i V_i(\theta_i)$ and $P_i(\theta_i)$ are the expected
utility and expected payment of user $i$ given his valuation report  $\theta_i$, respectively. Assume that $g$ is a nondecreasing function as in the literature of mechanism design. Intuitively, $g(\theta_i)$ is less than $\theta_i$ to give users incentives to truthfully report their $\theta_i$'s in the incomplete information scenario.
We let $g(\theta_{(i)})$ be the $i$th greatest among
  $g(\theta_1),\dots,g(\theta_n)$, then we have $g(\theta_{(1)})\geq\cdots\geq
  g(\theta_{(n)})$.
The following lemma helps simplify the constraints in
problem (\ref{eq:opt_Qb}).

\begin{proposition}[Necessary and sufficient for incentive compatibility]\label{lemma:optb}
  $V_i(\theta_i)$ is non-decreasing in $\theta_i$, and the differentiated pricing $P_i(\theta_i)$ is given by,
\begin{equation}\label{eq:P_Qb}
P_i(\theta_i)=\theta_i V_i(\theta_i)-\int_0^{\theta_i}V_i(\eta)d \eta.
\end{equation}
As a result, the platform's maximal profit, denoted by
$\Pi_D$, in
(\ref{eq:opt_Qb}) can be written as
\begin{equation}
\int \max_{m\in N} \bigg((1-\rho^m) \sum_{i=1}^{m}
  g(\theta_{(i)}) -mc\bigg) d F^n(\theta).\label{eq:mb}
\end{equation}
\end{proposition}
The proof is given in Appendix~\ref{lemma:optbproof}.
Proposition \ref{lemma:optb} indicates that at the optimum, the platform will include
$m$ users whose $g(\theta_i)$'s are the greatest.

\subsubsection{Uniform Pricing Scheme As Approximation}\label{UPSA}
Although the differentiated pricing mechanism in \eqref{eq:P_Qb} is
  optimal, it is complicated to compute and implement in practice. While
  it guarantees that truthful reporting is the best response for users ,  it is
  difficult for a user to check \eqref{constraint2b} for any $\theta_i^\prime$ and $\theta_{-i}$. Next, we propose a uniform pricing scheme  which does not even require
users to declare their $\theta$'s.

In this simple scheme,  the platform announces a
single price $P$ to users without any admission control.
As users are i.i.d. distributed, there is a common valuation threshold $\bar{\theta}$ for subscription decision-making and $\bar{\theta}$ depends on $P$.
User $i$ will decide subscription by comparing his $\theta_i$ to
$\bar{\theta}$ and participates if $\theta_i\geq \bar{\theta}$.
Approximately $m=n(1-F(\bar{\theta}))$ users will finally
subscribe and contribute to the network externalities.. User $i$'s payoff in \eqref{upb} becomes
\[
u_i=\theta_i \big(1-\rho^{n(1-F(\bar{\theta}))}\big)-P\geq 0,\quad\text{for all }
\theta_i\geq \bar{\theta}.
\]
This should be zero for an indifferent user with $\theta_i=\bar{\theta}$. Thus,
\begin{equation}\label{uniformprice}
P=\bar\theta \big(1-\rho^{n(1-F(\bar{\theta}))}\big),
\end{equation}
which is a function of $\bar{\theta}$, or we can equivalently express $\bar{\theta}$ as a function of $P$.
The platform's optimization problem is
\begin{equation}\label{opt:incom:unib}
\max_{\bar\theta}n(1-F(\bar{\theta}))\bar{\theta} \big(1-\rho^{n(1-F(\bar{\theta}))}\big)-n(1-F(\bar{\theta}))c.
\end{equation}

Since each $\theta_i$ follows the
uniform distribution on $[0, 1]$,  problem \eqref{opt:incom:unib} becomes
\begin{equation}\label{opt:incom:uni2b}
\max_{\bar\theta}n(1-\bar{\theta})\bar{\theta} \big(1-\rho^{n(1-\bar{\theta})}\big)-n(1-\bar{\theta})c
\end{equation}
 The uniform pricing problem (though non-convex) can be solved efficiently via an one-dimensional search. We next present the analytical results as  $n\rightarrow\infty$ and characterize the network value/profit.

 \begin{theorem}\label{opt:incom:asymb}
Given users' bounded utility model in \eqref{upb},  as $n\rightarrow\infty$, the optimal uniform price under incomplete information is $P^*\rightarrow\frac{1+c}{2}$, the optimal user threshold is $\bar\theta^*\rightarrow\frac{1+c}{2}$
and the maximal profit is $\Pi_U\sim(\tfrac{1-c}{2})^2n.$
 As $n\rightarrow\infty$, the maximum profit achieved by the differentiated pricing scheme in \eqref{eq:P_Qb} is $\Pi_D\sim(\tfrac{1-c}{2})^2n.$
Thus, uniform pricing is asymptotically optimal, i.e.,
$\lim\limits_{n\rightarrow\infty}\frac{\Pi_U}{\Pi_D}\rightarrow1.$
\end{theorem}

The proof is given in Appendix \ref{opt:incom:asymbproof}.  Theorem \ref{opt:incom:asymb} shows that uniform pricing scheme's profit grows at the same rate with $n$ as the differentiated pricing scheme.

\subsubsection{Price of Information}\label{sectionpoi}
 Now we are ready  to compare the expected maximal profits under complete and incomplete information.
 We define price of information (PoI) as the ratio of  the expected maximal profit under complete  and incomplete information as $n\rightarrow\infty$, i.e.,
 \begin{equation}\label{poidef}
 \PoI=\lim\limits_{n\rightarrow\infty}\frac{E_\theta(\Pi)}{\Pi_U}.
 \end{equation}
 PoI is of course
 similar in concept to the well-known idea of Price of
 Anarchy. However, the second refers to the social welfare that is
 lost when users act self-interestedly vis-a-vis for the community.
Note that price discounts are given as incentives under incomplete information, and the profit is greater under complete information. Thus, $\PoI\geq 1$.
One can also replace $\Pi_U$ by $\Pi_D$ in \eqref{poidef} without changing the PoI value, according to the uniform pricing's asymptotic optimality in Theorem \ref{opt:incom:asymb}.
 \begin{proposition}\label{poiexpoprop}
Given users' bounded utility model in \eqref{upb}, the price of information is $\PoI=2$.
\end{proposition}
The proof is given in Appendix \ref{poiexpopropproof}.
 We note that PoI does not depend on parameter $\rho$. Recall that $\rho$ tells the service coverage contributed by an individual user. As $n$ goes to infinity, the total  bounded coverage is fixed to 100\% , and hence $\rho$ has no impact on PoI .

\section{Optimal Pricing  For Unbounded Utility Model}\label{UUF}
In this section, we will analyse the platform's pricing strategy for unbounded user utility $\phi(m)=\log m$ .
 %function is given by  \eqref{uplog}.

\subsection{Pricing under Complete Information}
%Under complete information about all users' valuations $\theta_i$'s, the platform's optimization problem is to choose prices $p_i$'s and control admission $\pi_i$'s to maximize his profit.
Assume user's utility is given by \eqref{uplog}, which follows from Zipf's law.
Similar to Section \ref{combounded},
it is optimal to leave a zero payoff to user $i$
by setting the price to be
\begin{equation}%\label{eq:price_selection}
p_i^*(\pi)=\theta_i\log(\sum_{i}\pi_i) .\nonumber
\end{equation}
The platform's optimization problem is
\begin{equation}\label{opt:completeinfo_pi}
  \max_{\{\pi_i,\ {i\in{N}}\}}\ \sum_{i\in{N}}\pi_i \bigg(\theta_i\log (\sum_{j\in N} \pi_j)-c\bigg).
\end{equation}

 Lemma~\ref{lem:ij} still holds here, the problem
reduces to
\begin{equation}\label{opt:com:log}
  \max_{m\in N}\ \bigg(\log m\sum_{i=1}^m \theta_i - mc\bigg).
\end{equation}

Thus, similarly, the platform will
select $m$ users with the largest service valuations and we can use one-dimensional  search to find the optimal $m$.

\subsection{Pricings under Incomplete Information}
%Under incomplete information, the platform does not know $\theta_i$'s exactly but their distributions with cumulative distribution function $F$.
Inherit the same logic from Section~\ref{incombounded}, We will derive a optimal (differentiated) pricing
scheme and then propose a uniform pricing scheme as approximation.
We will compare these two different pricing schemes
asymptotically.
\subsubsection{Optimal/Differentiated Pricing Scheme}
The platform's optimization problem in differentiated pricing scheme is the same as
\eqref{eq:opt_Qb}-\eqref{constraint2b} except service value $(1-\rho^{\sum_j\pi_j(\theta_i,\theta_{-i})})$ is replaced by $\log(\sum_{j=1}^n \pi_j(\theta_i,\theta_{-i}))$.
 %Under incomplete information, the platform will require each user $i$ to declare his $\theta_i$. Given the $\theta_i$'s (may or may not be truthful) declared by the users, the platform should choose $p_i$'s and $\pi_i$'s as functions of the $\theta_i$'s distribution to maximize his profit, i.e.,
% Similar to Section \ref{ODPS}, the platform's optimization problem is
%\begin{align*}
%&\max_{\pi_i(\cdot),\ p_i(\cdot)} E_\theta \biggl(\sum_{i=1}^n \pi_i(\theta) (p_i(\theta)-c) \biggr) %\label{eq:opt_Q}
%\\
%&\text{subject to}\nonumber\\
%& E_{\theta_{-i}}\bigg(\pi_i(\theta_i,\theta_{-i}) \bigg(\theta_i \log(\sum_{j=1}^n %\pi_j(\theta_i,\theta_{-i}))-p_i(\theta_i,\theta_{-i})\bigg) \bigg)\geq 0,
%\label{constraint1}
%\\
%&E_{\theta_{-i}}\bigg(\pi_i(\theta_i,\theta_{-i})
%\bigg(\theta_i \log (\sum_{j=1}^n
%\pi_j(\theta_i,\theta_{-i}))-p_i(\theta_i,\theta_{-i})\bigg)
%\bigg)\nonumber \\
% & \geq E_{\theta_{-i}}\bigg(\pi_i(\theta_i',\theta_{-i}) \bigg(\theta_i \log(\sum_{j=1}^n \pi_j(\theta_i',\theta_{-i}))-p_i(\theta_i',\theta_{-i})\bigg) \bigg),
 %\label{constraint2}
% \\
%&\text{for\ all}\ i\ \text{and} \ \theta_i'.\nonumber
%\end{align*}
We can similarly define $g(\theta_i)$, $V_i(\theta_i)$, and $P(\theta_i)$ as in \eqref{g}, \eqref{V}, and \eqref{P}.
Then Proposition~\ref{lemma:optb} still holds here, and similar to \eqref{eq:mb}, the platform's maximal profit can be written as
\begin{equation}
\int \max_{m\in N} \bigg(\log m \sum_{i=1}^{m}
  g(\theta_{(i)}) -mc\bigg) \dif F^n(\theta).\label{eq:munbouned}
\end{equation}

This indicates that at the optimum, the platform will include
$m$ users whose $g(\theta_i)$'s are the greatest.

\subsubsection{Uniform Pricing Scheme as Approximation}
Now we analyse the uniform pricing mechanism. Similar to Section \ref{UPSA}, the payoff should be zero for an indifferent user with $\theta_i=\bar{\theta}$. Thus, similar to \eqref{uniformprice}, we have
\[
P=\bar\theta \log\big(n(1-F(\bar{\theta})\big),
\]
and the platform's optimization problem is
%\begin{equation}\label{opt:incom:uni}
%\max_{\bar\theta}n(1-F(\bar{\theta}))\bar{\theta} \log\big(n(1-F(\bar{\theta}))\big)-n(1-F(\bar{\theta}))c
%\end{equation}
%Since each $\theta_i$ has the
%uniform distribution on $[0, 1]$,  problem \eqref{opt:incom:uni} becomes
\begin{equation}\label{opt:incom:uni2}
\max_{\bar\theta}\bar{\theta}(1-\bar{\theta})n\log(n(1-\bar{\theta}))-n(1-\bar{\theta})c.
\end{equation}
The uniform pricing problem (though non-convex) can be
solved efficiently via an one-dimensional search.
We next
present the analytical results as $n\rightarrow\infty$ and characterize the
network value/profit.

 \begin{theorem}\label{opt:incom:uniinfty}
 Given users' unbounded utility model in \eqref{uplog}, as $n\rightarrow\infty$, the optimal uniform price under incomplete information is $P^*\rightarrow\tfrac{1}{2}\log \frac{n}{2}$, the optimal user threshold is $\bar\theta^*\rightarrow\tfrac{1}{2}$ and the
 maximal profit is
 $\Pi_U\sim\frac{n}{4}\log(\frac{n}{2}).$
As $n\rightarrow\infty$, the maximum profit achieved by the differentiated pricing scheme  is
$\Pi_D\sim\frac{n}{4}\log(\frac{n}{2}).$
Therefore, uniform pricing is asymptotically optimal, i.e.,
$\lim\limits_{n\rightarrow\infty}\frac{\Pi_U}{\Pi_D}\rightarrow1.$
\end{theorem}
The proof is given in Appendix \ref{opt:incom:uniinftyproof}. Note that the cost $c$ does not play a role in the optimal price or maximal profit. This is because when utility is unbounded, as $n\rightarrow\infty$, the user's perceived network value grows super-linearly with the number of participants, while the cost only grows linearly and is negligible.

 We next also consider Metcalfe's law rather than Zipf's law and $\phi(m)= m$ as a less concave function than $log(m)$ . Then user $i$'s payoff  is now given by \eqref{upu} and we can prove similar results as Theorem \ref{opt:incom:uniinfty}  below. The proof is given in Appendix \ref{asymptoticlinearproof}.
\begin{corollary}\label{asymptoticlinear}
Given users' unbounded utility model in \eqref{upu}, as $n\rightarrow\infty$, the optimal uniform price under incomplete information is $P^*\rightarrow(2/9)n$,  the optimal user threshold $\bar\theta^*=1/3$, and the maximal profit is
$\Pi_U\sim(4/27)n^{2}.$
As $n\rightarrow\infty$, the maximum profit achieved by the differentiated pricing scheme is $\Pi_D\sim(4/27)n^{2}.$
Therefore, uniform pricing is asymptotically optimal, i.e.,
$\lim\limits_{n\rightarrow\infty}\frac{\Pi_U}{\Pi_D}\rightarrow1.$
\end{corollary}

\subsubsection{Price of Information}

 We can still define price of information by  \eqref{poidef}.
 %It follows from Theorem \ref{opt:incom:uniinfty} and Corollary \ref{asymptoticlinear} that, given users' bounded model in \eqref{uplog} or \eqref{upu}, replacing $\Pi_U$ by $\Pi_D$ does not change the value of $\PoI$.
  %For tractability, we still assume that $\theta_i \sim U[0,1]$.
  We more generally consider  users' payoff function (not limited to \eqref{uplog} and \eqref{upu}) as follows,
  \begin{equation}
  u_i=\theta_iv(m)-p_i,\label{upgeneral}
  \end{equation}
  where $v(m)$ is an unbounded, increasing and concave function with  $v(0)=0$.
  Then we have the following proposition.
 \begin{proposition}\label{poiu}
 Given users' general unbounded utility model in \eqref{upgeneral}, the price of information is  $\PoI\in[2,27/8]$.
More specifically, if  users' utility model follows Zipf's law in \eqref{uplog}, $\PoI=2$. If users'  utility model follows Metcalfe's law in \eqref{upu}, $\PoI=27/8$.
\end{proposition}
The proof is given in Appendix \ref{poiuproof}. As the utility function becomes more concave (from $m$ in Metcalfe's law to $log(m)$ in Zipf's law), the profit loss due to lack of information decreases since the network externality  decreases and there is less consumer surplus to be transformed to platform's profit. This holds true for a general
cumulative distribution function $F$.

\section{Pricing  Extension to A Two-sided Market}\label{tsm}
In this section, we include another group/type of users to the platform, who are simply consumers and do not contribute to the network externalities .
Denote the set of original users (both contributors and consumers) as $N_1=\{1,2,\cdots,n_1\}$, and the new user set by $N_2=\{n_1+1,\cdots,n_1+n_2\}$. Within each set, we reorder users according to their service valuations such that $\theta_1>\cdots > \theta_{n_1}$ and $\theta_{n_1+1}> \cdots>\theta{n_1+n_2}$.
Note that  to which set a user belongs is public information as it is easy to verify whether a user can contribute or not. However, within each set, users' service valuations are still private information. As the two user sets' subscriptions affect each other, we wonder how the platform should jointly decide pricing schemes to the two sets of users. We also wonder if we can still approximate the two user groups' differentiated pricing via two uniform prices to achieve asymptotic optimality. The pricing schemes considered in Sections \ref{BUF} and \ref{UUF} can be similarly applied to the two sets of users. However, the asymptotical analysis becomes challenging as dimension increases.

Without much loss of generality, we apply  Metcalfe's law here, where  $\phi(m)=m$ and  $m$ only counts the original users in $N_1$ who can contribute.

%Then, user's utility, which is similar to \eqref{upu} with $b=1$, can be written down as
%\begin{equation}\label{uptwo}
%u_i=\theta_i\sum_{j\in N_1}\pi_j-p_i.
%\end{equation}
\subsection{Pricing under Complete Information  }
%\begin{align}
%&\max_{\{(\pi_i,\ p_i), {\ i\in{N}}\}}\ \sum_{j\in{N}}\pi_j (p_j-c)\nonumber\\
%&\text{subject\ to,\ } \pi_i\bigg(\theta_i\sum_{j\in N_1}\pi_j-p_i
%\bigg)\geq 0,\quad\text{for all } i\in{N},\label{opt:completeinfo2}
%\end{align}

Similar to Section \ref{combounded},
it is optimal  to leave a zero payoff to user $i$ of any user set
by setting the price to be
\begin{equation}%\label{eq:price_selection}
p_i^*(\pi)=\theta_i\sum_{j\in N_1}\pi_j.\nonumber
\end{equation}
 The platform's
optimization problem is
\begin{equation}\label{opt:completeinfo2_pi}
  \max_{\{\pi_i,\ {i\in{N}}\}}\ \sum_{i\in{N}}\pi_i \bigg(\theta_i\sum_{j\in N_1}\pi_j-c\bigg)\nonumber
\end{equation}
Similar to Lemma \ref{lem:ij},  at the optimality of problem \eqref{opt:completeinfo2_pi}, for any two users $i, j\in N_1$ or $i, j\in N_2$ with
  $\theta_i>\theta_j$, if user $j$ is included
 (i.e., $\pi_j=1$), then user $i$ should also be included
  ($\pi_i=1$).
It follows  that the platform will
select $m_1$ users with the largest service valuations in $N_1$ and $m_2$ users with the largest service valuations in $N_2$ and problem \eqref{opt:completeinfo2_pi}
reduces to
\begin{equation}\label{opt:CompInfo_two}
  \max_{m_1,m_2}\ m_1\bigg(\sum_{i=1}^{m_1} \theta_i+\sum_{i=n_1+1}^{n_1+m_2} \theta_i\bigg)-(m_1+m_2)c,
\end{equation}
which is an extension of \eqref{opt:com:log} for a single user set. We have the following theorem regarding the optimal solution to \eqref{opt:CompInfo_two}.
\begin{proposition}\label{two:cominfo:opt}
  Let $\bar m_2$ be the largest user number $m_2$ such that $n_1\theta_{n_1+m_2}\ge c$. Then if
  \[\bigg(\sum_{i=1}^{n_1} \theta_i+\sum_{i=n_1+1}^{n_1+\bar m_2} \theta_i\bigg)-\frac{n_1+\bar m_2}{n_1}c>0,\]
  then the optimal solution to \eqref{opt:CompInfo_two} is $m_1^*=n_1$ and $m_2^*=\bar m_2$. Otherwise, the optimal solution to \eqref{opt:CompInfo_two} is $m_1^*=0$ and $m_2^*=0$.
\end{proposition}
The proof is given in Appendix \ref{two:cominfo:optproof}. It is optimal to either include all the potential contributors in the platform for the maximum network externality  or include no users due to high cost. Note that if no user of the first set is selected, the network value is zero and the platform cannot attract any pure user from the other set.
\subsection{Pricing under Incomplete information}
%Under incomplete information, the platform does not know $\theta_i$'s exactly but their distributions. This distribution is known to all users and the platform. We assume $\theta_i$'s are independent and identically distributed samples from a distribution on $[0,1]$ with cumulative distribution function $F$. We will derive a differentiated pricing scheme and a uniform pricing scheme for the platform. We will compare these two different pricing schemes asymptotically.
\subsubsection{Differentiated Pricing Scheme}
 Under incomplete information, the platform will require each user $i$ of each type to declare his $\theta_i$ and then choose $p_i$'s and $\pi_i$'s to maximize its profit.
 We can similarly decide the differentiated pricing in \eqref{eq:P_Qb}  as  Proposition~\ref{lemma:optb} still applies, and the platform's
optimization problem can be written as
\begin{align*}
\Pi = &\int \max_{m_1,m_2}  m_1\bigg(\sum_{i=1}^{m_1}
  g(\theta_{(i)}) +\sum_{i=n_1+1}^{n_1+m_2}
  g(\theta_{(i)})\bigg) \\
  &-(m_1+m_2)c\,\dif F^n(\theta).
\end{align*}
%It is difficult to jointly optimize $m_1$ and $m_2$. There are no closed-form solutions.
\subsubsection{Uniform Pricing Scheme as Approximation}
Unlike the single user type case, in the two-sided market,  the platform sets different uniform prices for different user types. The price for type-1 users (dual-role) is $P_1$ and and the price for type-2 user (pure consumers) is $P_2$.
There is a unique   threshold for each type of users: $\bar\theta_1$ for type-1 and $\bar\theta_2$ for type 2 . Assume $\theta_i$ is uniformly distributed in $[0,1]$.
Similar to \eqref{uniformprice}, for a type-1 user $i\in N_1$ with  $\theta_i=\bar\theta_1$, we have
$$\bar\theta_1(1-\bar\theta_1)n_1=P1.$$
and for a type 2 user user $i\in N_2$ with $\theta_i=\bar\theta_2$, we have
$$\bar\theta_2(1-\bar\theta_1)n_1=P2.$$
The platform's
optimization problem can be written as
\begin{align*}
\max_{\bar\theta_1,\bar\theta_2\in[0,1]}&n_1(1-\bar\theta_1)(\bar\theta_1(1-\bar\theta_1)n_1)n_2-c)\\
&+n_2(1-\bar\theta_2)(\bar\theta_2(1-\bar\theta_1)n_1-c).
\end{align*}
Assume $n_1/n_2=k$ where $k$ is a positive constant,  when $n_1$ and $n_2$ or simply $n$ go to infinity,
%\begin{align*}
%\max_{\bar\theta_1,\bar\theta_2\in[0,1]}&kn_2(1-\bar\theta_1)(\bar\theta_1(1-\bar\theta_1)kn_2)n_2)\\
%&+n_2(1-\bar\theta_2)(\bar\theta_2(1-\bar\theta_1)kn_2).
%\end{align*}
we have the following proposition regarding the optimal uniform prices and maximal profits.

  \begin{theorem}\label{two:UniPri}
In two-sided market, as $n\rightarrow\infty$,  the two
 optimal uniform prices under incomplete information
are
\begin{equation*}
P_1=\begin{cases}
0&\mbox{ if } \frac{n_1}{n_2}<\frac{1}{4},\\
\frac{4n_1-3n_2+2\sqrt{4n_1^2+3n_1n_2}}{36}&\mbox{ if } \frac{n_1}{n_2}\ge \frac{1}{4},
\end{cases}
\end{equation*}
\begin{equation*}
P_2=\begin{cases}
\frac{n_1}{2}&\mbox{ if } \frac{n_1}{n_2}<\frac{1}{4},\\
\frac{n_1n_2}{4\sqrt{4n_1^2+3n_1n_2}-8n_1}&\mbox{ if } \frac{n_1}{n_2}\ge \frac{1}{4},
\end{cases}
\end{equation*}
yielding the optimal profit
\begin{equation}\label{maxprofit:two}
\Pi_U=\begin{cases}
\frac{4n_1^2+6n_2^2+7n_1n_2+(4n_1+3n_2)\sqrt{4n_1^2+3n_1n_2}}{108}&\mbox{ if } \frac{n_1}{n_2}<\frac{1}{4},\\
\frac{1}{4}n_1n_2&\mbox{ if } \frac{n_1}{n_2}\ge \frac{1}{4}.
\end{cases}
\end{equation}
 %$$\bar\theta_1^*=\frac{2}{3}-\frac{\sqrt{k(4k+3)}}{6k},\qquad\bar\theta_2^*=\frac{1}{2}.$$
% $$\bar\theta_1^*=0,\qquad \bar\theta_2^*=\frac{1}{2}.$$
$P_1^*$ decreases with  $n_1/n_2$, $P_2^*$ increases with $n_1/n_2$, and $\Pi_U$ increases with  $n_1/n_2$.
As $n\rightarrow\infty$, the profit achieved by the differentiated pricing scheme is $\Pi_D\sim\eqref{maxprofit:two}$, that is,
uniform pricing is asymptotically optimal.
\end{theorem}
The proof is given in Appendix \ref{two:UniPriproof}.
When $n_1/n_2$ is small (less than 1/4), the platform platform's profit comes mostly from the type-2 pure users and desires the maximum network externalities contributed by the type-1 users. Thus, it
charges zero price to motivate all type-1 users to contribute to the network externalities.
As $n_1/n_2$ increases, the fraction of potential contributors increases, the platform with larger network externalities can charge more from the pure users of type-2, while keeping more contributors of type-1 at a lower price. Thus, $P_1^*$ decreases and $P_2^*$ increases with $n_1/n_2$.

\subsubsection{Price of Information}

%Since $\frac{\Pi_U}{\Pi_D}\rightarrow1$ as $n_2\rightarrow\infty$ ,we can still define price of information by \eqref{poidef}
%We expect the profit under complete information is larger, i.e., $\PoI\le1$. For tractability, we still assume that $\theta_i \sim U[0,1]$. Then
Similar to Section \ref{sectionpoi}, we can define price of information by \eqref{poidef} and straightforward calculation gives
the following theorem.
\begin{theorem}\label{two:poi}
In the two-sided market,
 the price of information is
\[\label{maxprofit}
\PoI=\begin{cases}
\frac{54(n_1^2+n_1n_2)}{4n_1^2+6n_2^2+7n_1n_2+(4n_1+3n_2)\sqrt{4n_1^2+3n_1n_2}}&\mbox{ if } \frac{n_1}{n_2}<\frac{1}{4},\\
\frac{2(n_1+n_2)}{n_2}&\mbox{ if } \frac{n_1}{n_2}\ge \frac{1}{4}.
\end{cases}
\]

Overall, price of information decreases as $n_1/n_2$ increases.
\end{theorem}
The proof is given  in Appendix \ref{two:poiproof}. As $n_1/n_2$ increases, the fraction of potential contributors increases, the platform under incomplete information still needs to provide price discounts as incentives. As a result, the PoI or profit loss due to lack of information increases.

\section{Simulation Results}\label{simulation}
%\subsection{Accuracy of Asymptotic Approximation}
%Figure~\ref{fig:ProfitRatioB} shows that the average profit ratio between the uniform and differentiated pricing mechanisms  and price of information when given bounded utility model and users' valuations follow the uniform distribution. In this figure, the average profit under different pricing schemes are obtained by averaging 1 million sample data  with different $\theta$ realizations. We observe as the $n$ increases, the gap between the two cases decreases  and the uniform pricing will perform very close to the differentiated pricing. We also observe as $b$ increases, price of information decreases.
%\begin{figure}[hh]
%\centering
%\subfigure[Bounded Utility Model]{
%\begin{minipage}[t]{0.46\linewidth}\label{fig:ProfitRatioB}
%\centering
%\includegraphics[width=1.6in]{profitratiobounded.pdf}
%\end{minipage}
%}
%\subfigure[Unbounded Utility Model]{
%\begin{minipage}[t]{0.46\linewidth}\label{fig:ProfitRatioU}
% \centering
%\includegraphics[width=1.6in]{profitratiounbounded.pdf}
%\end{minipage}
%\caption{Average profit ratio between uniform and differentiated pricing and price of information.}
%\label{profitratio}
%\end{figure}
\begin{figure}[hh]
	\centering
	\includegraphics[width=2in]{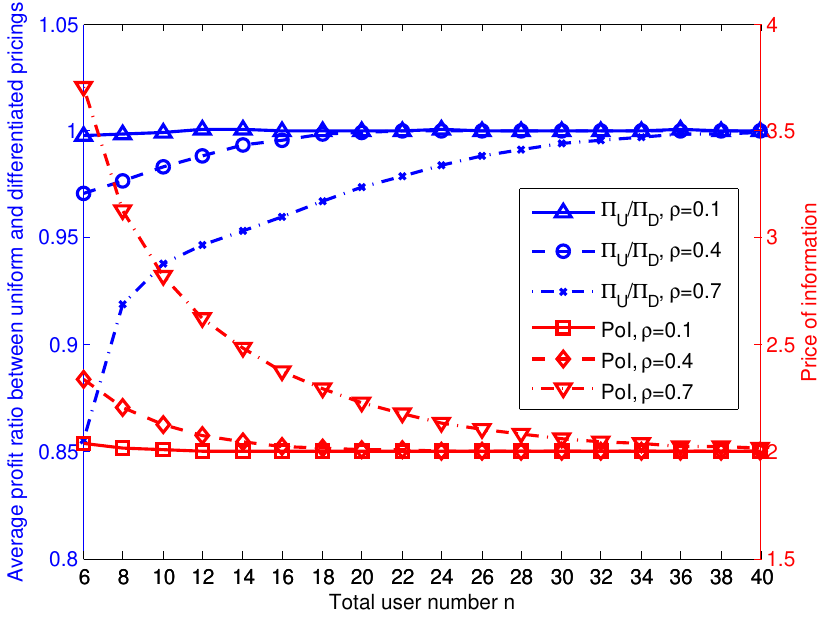}
	\caption{Average profit ratio between uniform and differentiated pricing and price of information for bounded utility model.} \label{fig:ProfitRatioB}
\end{figure}

We plot ratios of average profits under different pricing schemes in Figure~\ref{fig:ProfitRatioB}.\ for bounded utility model and Figure~\ref{fig:ProfitRatioU}.\ for unbounded utility model, by averaging 1 million sample data  with different $\theta$ realizations.

Figure~\ref{fig:ProfitRatioB}.\ shows that the average profit ratio
between the uniform and differentiated pricing schemes  increases with user number. This is consistent with Theorem~\ref{opt:incom:asymb}, which shows uniform pricing scheme is asymptotically optimal as $n$ goes to infinity. The convergence rate at which $\Pi_U/\Pi_D$ approaches 1 decreases with $\rho$. As service
coverage contributed by an individual user increases ($\rho$ decreases), total service converges to 100\% faster and hence uniform pricing scheme approaches  optimality faster.
Figure~\ref{fig:ProfitRatioB}.\ also shows PoI as an decreasing function of $n$, approaches to 2 as in Proposition~\ref{poiexpoprop}.
PoI decreases with $n$ since the information of users' valuation distribution helps pricing design of the platform more as $n$ increases.
The convergence rate at which PoI approaches 2 increases as $\rho$ decreases is also due to the fact that total service converges to 100\% faster as $\rho$ decreases.

\begin{figure}[hh]
	\centering
	\includegraphics[width=2in]{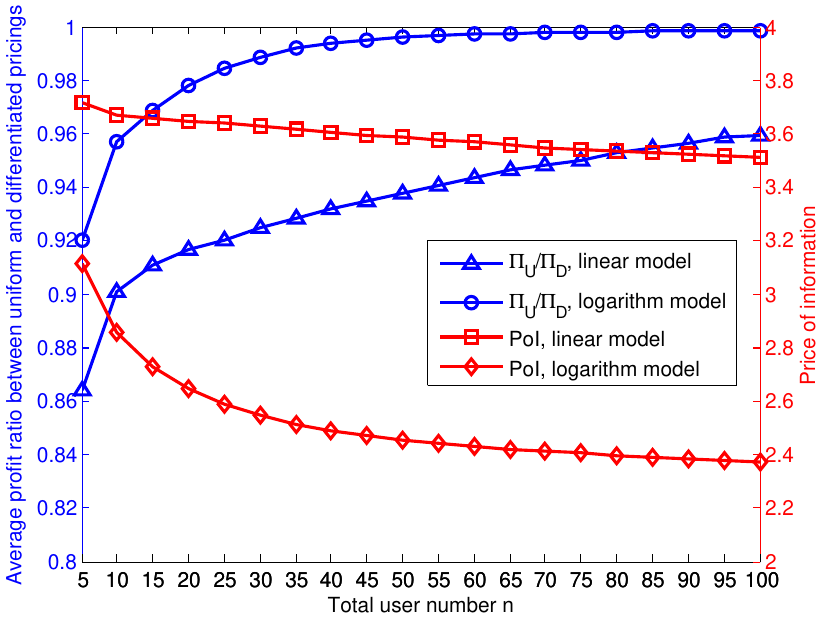}
	\caption{Average profit ratio between uniform and differentiated pricing and price of information for unbounded utility model.} \label{fig:ProfitRatioU}
\end{figure}

Figure~\ref{fig:ProfitRatioU}.\ shows that the average profit ratio
between the uniform and differentiated pricing mechanisms  increases with user number. This is consistent with Theorem~\ref{opt:incom:uniinfty} and Corollary~\ref{asymptoticlinear}, which show uniform pricing scheme is asymptotically optimal as $n$ goes to infinity. Logarithm utility model converges faster than linear utility model.
This is because network externalities grow faster in linear utility model than logarithm utility model and hence
uniform pricing cause greater profit loss in linear utility model than logarithm utility model.
Figure~\ref{fig:ProfitRatioU}.\ also shows PoI as an decreasing function of $n$ approaches to 2 for logarithm utility model  and 27/8 (=3.375) for linear utility model. This is consistent with Proposition~\ref{poiu}.

\section{Conclusion}\label{conclusion}
  This paper studies how a peer-to-peer sharing platform should price its service  to maximize its profit.
  We consider both bounded and unbounded user utility models.
  For both bounded and unbounded user utility models, we analyze the optimal pricing schemes to select heterogeneous users in the platform under complete and incomplete information of users' service valuations.  The  profit loss  due to lack of information becomes greater as the utility function becomes less concave. We show that the complicated differentiated pricing scheme under incomplete information can be replaced  by a single uniform price with asymptotic optimality. We also extend our pricing schemes  to a  two-sided market. Platform may charge zero price to the original group of users in order to attract the pure user group. Uniform pricing scheme is still asymptotically optimal as user number goes to infinity and price of information increases as the fraction of original users decreases.

\newpage

\appendix
\section{Proof of Proposition \ref{lemma:optb}}\label{lemma:optbproof}
\begin{proof}
Given \eqref{g}, \eqref{V}, and \eqref{P}, incentive compatibility constraint \eqref{constraint2b} can be rewritten as follows,
\begin{equation}\label{VP}
\theta_iV_i(\theta_i)-P_i(\theta_i)\ge\theta_iV_i(\theta_i^\prime)-P_i(\theta_i^\prime)
\end{equation}
for all $i$ and $\theta_i^\prime\in[0,1]$.
Assume $V_i(\theta_i)$ is non-decreasing in $\theta_i$, and
\[
P_i(\theta_i)=\theta_i V_i(\theta_i)-\int_0^{\theta_i}V_i(\eta)d \eta.
\]
Straightforward calculation shows that \eqref{VP} is satisfied.

Now assume \eqref{VP} holds for all $i$ and $\theta_{i}^\prime\in[0,1]$. Fix arbitrary $i$,  \eqref{VP} holds
for any $\theta_i\in[0,1]$ and $\theta_i^\prime\in[0,1]$. It follows that, for any $x,y\in[0,1]$
\begin{equation}\label{xy}
xV_i(x)-P_i(x)\ge xV_i(y)-P_i(y),
\end{equation}
\begin{equation}\label{yx}
yV_i(y)-P_i(y)\ge yV_i(x)-P_i(x).
\end{equation}
Adding \eqref{xy} and \eqref{yx} yields
\[(x-y)(V_i(x)-V_i(y))\ge0.\]
Thus, $V_i(\theta_i)$ is non-decreasing in $\theta_i$. Assume $x\ge y$, by definition of Riemann integral, it follows from \eqref{xy} that
\[xV_i(x)-\int_{0}^{x}V_i(\eta)\dif \eta\ge P_i(x),\]
it follows from \eqref{yx} that
\[-xV_i(x)+\int_{0}^{x}V_i(\eta)\dif \eta\ge -P_i(x).\]
Then,
\[
P_i(\theta_i)=\theta_i V_i(\theta_i)-\int_0^{\theta_i}V_i(\eta)d \eta.
\]
Thus, the first part of the theorem follows.

Using (\ref{eq:P_Qb}), the platform's maximal profit
$\Pi_D$ can be written as
\begin{align}
\Pi_D &=  \max_{\pi_i(\cdot),\ p_i(\cdot)} E_\theta \bigg(\sum_{i=1}^n \pi_i(\theta) (p_i(\theta)-c) \bigg)\nonumber\\
  &=\max_{\pi_i(\cdot)}\int \sum_{i=1}^n \pi_i(\theta) \bigg(g(\theta_i)  (1-\rho^{\sum_{j}\pi_j(\theta)}) - c \bigg) \dif F^n (\theta)\nonumber \\
  &=\max_{m(\cdot)} \int \sum_{i=1}^{m(\theta)} \bigl(g(\theta_{(i)}) (1-\rho^{m(\theta)})  -m(\theta)c\bigr) \dif F^n(\theta)\nonumber\\
  &=\max_{m(\cdot)} \int \bigg((1-\rho^{m(\theta)}) \sum_{i=1}^{m(\theta)}
  g(\theta_{(i)}) - m(\theta)c\bigg) \dif F^n(\theta)\nonumber\\
  &= \int \max_{m\in N} \bigg((1-\rho^m)\sum_{i=1}^{m}
  g(\theta_{(i)}) -m c\bigg) \dif F^n(\theta).\nonumber%\label{eq:m}
\end{align}
Thus, the second part of the theorem follows.
\end{proof}
\section{Proof of Theorem \ref{opt:incom:asymb}}\label{opt:incom:asymbproof}
We first prove some lemmas.
\begin{lemma}\label{orderstat}
  Let $\theta_1,\dotsc,\theta_n$ be i.i.d.\ $U[0,1]$ and let
  $\theta_{(1)},\dotsc,\theta_{(n)}$ be their order statistics,
  $\theta_{(1)}\geq\cdots\geq \theta_{(n)}$. Then
\begin{align*}
   E[\theta_{(i)}] &= 1-\frac{i}{n+1},\\[4pt]
   \var[\theta_{(i)}] &=\frac{i (n+1-i)}{(n+1)^2 (n+2)},\\[4pt]
 \cov[\theta_{(i)},\theta_{(j)}] &=\frac{i(n+1-j)}{(n+1)^2 (n+2)},
 \quad i<j.
\end{align*}
\end{lemma}
The proof is by calculation with the joint density function for
$(\theta_{(i)},\theta_{(j)})$. The import of the next lemma is that for large $n$ the optimal number
of users who will participate is with high probability close to $((1-c)/2) n$.

\begin{lemma}\label{lemmaxb}\ \\
(i) Let $m=an$, $0<a<1$. As $n\rightarrow\infty$,
\begin{align*}
E(g(m))&\sim1-2a,\\
\var(g(m))&\sim\frac{4a(1-a)}{n}.
\end{align*}
(ii)  Let
\[m(\theta)=\arg\max_{m\in N}(1-\rho^m)\sum_{i=1}^{m}
  g(\theta_{(i)})-mc.\]
For any $\epsilon>0$, there is a large enough $n$ such that
\begin{align*}
 m(\theta)< \left(\tfrac{1-c}{2}-\epsilon\right)n&\implies g(\theta_{\left(\left(\tfrac{1-c}{2}-\epsilon\right)n\right)})<c+\epsilon,\\
m(\theta)>  \left(\tfrac{1-c}{2}+\epsilon\right)n &\implies g(\theta_{\left(\left(\tfrac{1-c}{2}+\epsilon\right)n\right)})>c-\epsilon .
\end{align*}
(iii) As $n\rightarrow\infty$,
\begin{align*}
P\bigg( m(\theta)< \left(\tfrac{1-c}{2}-\epsilon\right)n\bigg)&\sim O\left(1/n\right),\\
P\bigg(m(\theta)>  \left(\tfrac{1-c}{2}+\epsilon\right)n \bigg)&\sim O\left(1/n\right).
\end{align*}
\end{lemma}
\begin{proof}
  For (i), directly compute $E(g(m))$ and $\var(g(m))$ according to Lemma \ref{orderstat}, we get
  \begin{align*}
E(g(m))&\sim1-2a,\\
\var(g(m))&\sim\frac{4a(1-a)}{n}.
  \end{align*}

For (ii), to prove
\[m(\theta)< (\tfrac{1-c}{2}-\epsilon)n\implies g(\theta_{((\tfrac{1-c}{2}-\epsilon)n)})<c+\epsilon,\]
we will prove its equivalent statement
 \[g(\theta_{((\tfrac{1-c}{2}-\epsilon)n)})\ge c+\epsilon\implies m(\theta)\ge (\tfrac{1-c}{2}-\epsilon)n.\]
Suppose it is true that
$g(\theta_{((\tfrac{1-c}{2}-\epsilon)n)})\ge c+\epsilon.$
For any $0< m<(\tfrac{1-c}{2}-\epsilon)n$, if the platform user number increases from $m$ to $(\tfrac{1-c}{2}-\epsilon)n$,  the increment of platform's profit is
\begin{align*}
&(1-\rho^{(\tfrac{1-c}{2}-\epsilon)n})\sum_{i=1}^{(\tfrac{1-c}{2}-\epsilon)n}g(\theta_{(i)}) -(\tfrac{1-c}{2}-\epsilon)nc\\
&-(1-\rho^m)\sum_{i=1}^{m} g(\theta_{(i)}) +m c\\
&\geq(1-\rho^{(\tfrac{1-c}{2}-\epsilon)n})((\tfrac{1-c}{2}-\epsilon)n-m)(c+\epsilon) \\ &\quad-((\tfrac{1-c}{2}-\epsilon)n-m)c\\
&>0\,,
\end{align*}
since $n\rightarrow\infty$. Thus, it follows that $m(\theta)>(\tfrac{1-c}{2}-\epsilon)n$.

Similarly, to prove
\[m(\theta)>  (\tfrac{1-c}{2}+\epsilon)n \implies g(\theta_{((\tfrac{1-c}{2}+\epsilon)n)})>c-\epsilon \]
we will prove its equivalent statement
 \[g(\theta_{((\tfrac{1-c}{2}+\epsilon)n)})\le c-\epsilon\implies m(\theta)\le  (\tfrac{1-c}{2}+\epsilon)n .\]
Suppose it is true that $g(\theta_{((\tfrac{1-c}{2}+\epsilon)n)})\le c-\epsilon$.
For any $(\tfrac{1-c}{2}+\epsilon)n<m\le n$ , if the platform user number increases from $m-1$ to $m$,  the increment of platform's profit is
\begin{align*}
h(m)=\,&\rho^{m-1}(1-\rho)\sum_{i=1}^{m-1} g(\theta_{(i)})+(1-\rho^m)g(\theta_{(m)})-c\\
\le\,&\rho^{m-1}(1-\rho)m+(1-\rho^m)(c-\epsilon)-c\\
<\,&0,
  \end{align*}
since $n\rightarrow\infty$. Thus, it follows that $m(\theta)\le  (\tfrac{1-c}{2}+\epsilon)n$

 For (iii) we use Chebyshev's inequality. By (i) and (ii),
 \begin{align*}
P\bigg(m(\theta)< (\tfrac{1-c}{2}-\epsilon)n\bigg)
&\le P\bigg(g(\theta_{((\tfrac{1-c}{2}-\epsilon)n)})<c+\epsilon\bigg)\\
&\le\frac{\var[g(\theta_{((\tfrac{1-c}{2}-\epsilon)n)})-c-\epsilon]}{E^2[g(\theta_{((\tfrac{1-c}{2}-\epsilon)n)})-c-\epsilon]}\\
&\sim O(1/n).
\end{align*}
Similarly, by (i)  and (ii),
 \begin{align*}
P\bigg(m(\theta)> (\tfrac{1-c}{2}+\epsilon)n\bigg)
&\le P\bigg(g(\theta_{((\tfrac{1-c}{2}+\epsilon)n)})>c-\epsilon\bigg)\\
&\le\frac{\var[g(\theta_{((\tfrac{1-c}{2}-\epsilon)n)})-c+\epsilon]}{E^2[g(\theta_{((\tfrac{1-c}{2}-\epsilon)n)})-c+\epsilon]}\\
&\sim O(1/n).\qedhere
\end{align*}

\end{proof}

Now we will prove the theorem.
\begin{proof}
When $n\rightarrow\infty$, the platform's uniform pricing problem becomes
\[\max_{\bar\theta}n(1-\bar{\theta})\bar{\theta} -n(1-\bar{\theta})c.\]
It's a quadratic function and the optimal threshold is $\bar{\theta}^*=\frac{1+c}{2}$ and the resulting maximal profit is $(\tfrac{1-c}{2})^2n$.

Note that $\Pi_D\ge\Pi_U$, we only need to show that $\Pi_D$ is bounded by $(\tfrac{1-c}{2})^2n$.
We use below that $g\le 1$.
\begin{align*}
\int  & \max_{m\in N}(1-\rho^m)\bigg(\sum_{i=1}^{m}
  g(\theta_{(i)})\bigg)-mc\,\dif F^n(\theta)\\
&=
E\left[\max_{m\in N} (1-\rho^m) \bigg(\sum_{i=1}^{m}
  g(\theta_{(i)})\bigg)-mc\right]\\
&\leq P\Bigl(|m(\theta)-\tfrac{1-c}{2}n|>\epsilon n\Bigr)n\\
&\quad+ E\left[\max_{m:|m- \tfrac{1-c}{2}n|\leq \epsilon n} (1-\rho^m)\bigg(\sum_{i=1}^{m}
  g(\theta_{(i)})\bigg)\right].
\end{align*}

By Lemma \ref{lemmaxb}(iii), the first part is $\sim O(1)$. Now we compute the second fart as follows,
\begin{align*}
 &E\left[\max_{m:|m- \tfrac{1-c}{2}n|\leq \epsilon n} (1-\rho^m)\bigg(\sum_{i=1}^{m}
  g(\theta_{(i)})\bigg)\right]\\
  %\le &E\left[\max_{m:|m- \tfrac{1-c}{2}n|\leq \epsilon n}  \bigg(\sum_{i=1}^{(\tfrac{1-c}{2}-\epsilon)n}
  %g(\theta_{(i)})+\sum_{i=1+(\tfrac{1-c}{2}-\epsilon)n}^{m}g(\theta_{(i)})\bigg)\right]\\
  &\leq E\left[\sum_{i=1}^{(\tfrac{1-c}{2}-\epsilon)n}
  g(\theta_{(i)})+2\epsilon n\right].
\end{align*}
Using Lemma \ref{orderstat}, we see that the second part  is  $\sim(\tfrac{1-c}{2})^2n$. Therefore, the profit achieved by the differentiated pricing scheme is $\sim(\tfrac{1-c}{2})^2n$.
\end{proof}

\section{Proof of Proposition \ref{poiexpoprop}}\label{poiexpopropproof}
 \begin{proof}
 As $n\rightarrow\infty$
\begin{align*}
E_\theta(\Pi) &=   E_\theta \Bigg(  \max_{\{\pi_i,\ {i\in{N}}\}}\ \sum_{i\in{N}}\pi_i \bigg(\theta_i(1-\rho^{\sum_{i}\pi_i})-c\bigg)\Bigg)\nonumber\\
  &\sim  E_\theta \Bigg(\max_{\{\pi_i,\ {i\in{N}}\}}\ \sum_{i\in{N}}\pi_i \bigg(\theta_i-c\bigg)\Bigg)\notag\\
  &= E_\theta \Bigg(\sum_{i:\{\theta_i\ge c\}} (\theta_i - c)\Bigg)\notag\\
  &=n\int_c^1(\theta-c)\dif F(\theta)\\
  &=\frac{(1-c)^2}{2}n.
\end{align*}
Thus,
\[\PoI=\lim\limits_{n\rightarrow\infty}\frac{E_\theta(\Pi)}{\Pi_U}=\frac{\frac{(1-c)^2}{2}n}{(\tfrac{1-c}{2})^2n}=2.\qedhere\]
 \end{proof}

\section{Proof of Theorem \ref{opt:incom:uniinfty}}\label{opt:incom:uniinftyproof}
We first prove some lemmas.
\begin{lemma}\label{lemmalog}\ \\
(i) Let
\[m(\theta)=\arg\max_{m\in N}\log (m) \bigg(\sum_{i=1}^{m}
  g(\theta_{(i)})\bigg)-mc.\]
Then for any $\epsilon>0$, there is a large enough $n$ such that
\begin{align*}
 m(\theta)< (\tfrac{1}{2}-\epsilon)n&\implies g(\theta_{((\tfrac{1}{2}-\epsilon)n)})<\epsilon,\\
m(\theta)>  (\tfrac{1}{2}+\epsilon)n &\implies g(\theta_{((\tfrac{1}{2}+\epsilon)n)})>-\epsilon .
\end{align*}

(ii) As $n\rightarrow\infty$,
\begin{align*}
P\bigg(m(\theta)< (\tfrac{1}{2}-\epsilon)n\bigg)&\sim O(1/n),\\
P\bigg(m(\theta)>  (\tfrac{1}{2}+\epsilon)n\bigg)&\sim O(1/n).
\end{align*}
\end{lemma}
\begin{proof}
For (i), to prove
\[m(\theta)< (\tfrac{1}{2}-\epsilon)n\implies g(\theta_{((\tfrac{1}{2}-\epsilon)n)})<\epsilon,\]
we will prove its equivalent statement
 \[g(\theta_{((\tfrac{1}{2}-\epsilon)n)})\ge\epsilon\implies m(\theta)\ge (\tfrac{1}{2}-\epsilon)n.\]
Suppose it is true that $g(\theta_{((\tfrac{1}{2}-\epsilon)n)})\ge\epsilon$.
For any $0< m<(\tfrac{1}{2}-\epsilon)n$, if the subscriber number increases from $m$ to $(\tfrac{1}{2}-\epsilon)n$,  the increment of platform's profit is
\begin{align*}
&\log((\tfrac{1}{2}-\epsilon)n) \sum_{i=1}^{(\tfrac{1}{2}-\epsilon)n}g(\theta_{(i)}) -(\tfrac{1}{2}-\epsilon)nc\\
&-\log(m)  \sum_{i=1}^{m} g(\theta_{(i)}) +m c\\
&=\bigg(\log((\tfrac{1}{2}-\epsilon)n)-\log(m)\bigg)\sum_{i=1}^{m}g(\theta_{(i)}) \\ &\quad+\log((\tfrac{1}{2}-\epsilon)n)\sum_{i=m +1}^{(\tfrac{1}{2}-\epsilon)n} g(\theta_{(i)})-((\tfrac{1}{2}-\epsilon)n-m)c\\
&>((\tfrac{1}{2}-\epsilon)n-m)\bigg(\epsilon\log((\tfrac{1}{2}-\epsilon)n)-c\bigg)\\
&>0,
\end{align*}
since $n\rightarrow\infty$. Thus, it follows that $m(\theta)\ge (\tfrac{1}{2}-\epsilon)n$.

Similarly, to prove
\[m(\theta)> (\tfrac{1}{2}+\epsilon)n \implies g(\theta_{((\tfrac{1}{2}+\epsilon)n)})>-\epsilon ,\]
we will prove its equivalent statement
 \[g(\theta_{((\tfrac{1}{2}+\epsilon)n)})\le-\epsilon\implies m(\theta)\le (\tfrac{1}{2}+\epsilon)n .\]
Suppose it is true that $g(\theta_{((\tfrac{1}{2}+\epsilon)n)})\le-\epsilon$.
For any $(\tfrac{1}{2}+\epsilon)n<m\le n$ , if the subscriber number increases from $(\tfrac{1}{2}+\epsilon)n$ to $m$,  the increment of platform's profit is
\begin{align*}
&\log(m) \sum_{i=1}^{m} g(\theta_{(i)}) -m c-\log((\tfrac{1}{2}+\epsilon)n) \sum_{i=1}^{(\tfrac{1}{2}+\epsilon)n}g(\theta_{(i)}) \\
&+(\tfrac{1}{2}-\epsilon)nc\\
&=\bigg(\log(m)-\log((\tfrac{1}{2}+\epsilon)n)\bigg)\sum_{i=1}^{(\tfrac{1}{2}+\epsilon)n}g(\theta_{(i)}) \\
&\quad+\log(m)\sum_{i=(\tfrac{1}{2}-\epsilon)n +1}^{m} g(\theta_{(i)})-(m-(\tfrac{1}{2}+\epsilon)n )c\\
&<\bigg(\log(m)-\log((\tfrac{1}{2}+\epsilon)n)\bigg)(\tfrac{1}{2}+\epsilon)n\\ &\quad-\log(m)(m-(\tfrac{1}{2}-\epsilon)n)\epsilon\\
&<0,
\end{align*}
since $n\rightarrow\infty$. Thus, it follows that $ m(\theta)\le (\tfrac{1}{2}+\epsilon)n$

 For (ii) we use Chebyshev's inequality. By (i)  and Lemma~\ref{lemmaxb}(i),
 \begin{align*}
P\bigg(m(\theta)< (\tfrac{1}{2}-\epsilon)n\bigg)&\le P\bigg(g(\theta_{((\tfrac{1}{2}-\epsilon)n)})-\epsilon<0\bigg)\\
&\le \frac{\var[g(\theta_{((\tfrac{1}{2}-\epsilon)n)})-\epsilon]}{E^2[g(\theta_{((\tfrac{1}{2}-\epsilon)n)})-\epsilon]}\\
&\sim O(1/n).
\end{align*}
Similarly, by (i)  and Lemma~\ref{lemmaxb}(i),
 \begin{align*}
  P\bigg(m(\theta)>  (\tfrac{1}{2}+\epsilon)n\bigg)&\le P\bigg(g(\theta_{((\tfrac{1}{2}+\epsilon)n)})+\epsilon>0\bigg)\\
  &\le \frac{\var[g(\theta_{((\tfrac{1}{2}+\epsilon)n)})+\epsilon]}{E^2[g(\theta_{((\tfrac{1}{2}+\epsilon)n)})+\epsilon]}\\
&\sim O(1/n).\qedhere
\end{align*}
\end{proof}

Now we will prove the theorem.
\begin{proof}
When $n\rightarrow\infty$, the platform's uniform pricing problem becomes
\[\max_{\bar\theta}\bar{\theta}\log(n(1-\bar{\theta}) ).\]
The optimal threshold is $\bar{\theta}^*=\tfrac{1}{2}$ and the resulting maximal profit is $\frac{n}{4}\log(\frac{n}{2})$.

Note that $\Pi_D\ge\Pi_U$, we only need to show that $\Pi_D$ is bounded by $\frac{n}{4}\log(\frac{n}{2})$.
We use below that $g\le 1$.
\begin{align*}
 \int  &\max_{m\in N}\log (m) \bigg(\sum_{i=1}^{m}
  g(\theta_{(i)})\bigg)-mc\,\dif F^n(\theta)\\
&=
E\left[\max_{m\in N} \log (m) \bigg(\sum_{i=1}^{m}
  g(\theta_{(i)})\bigg)-mc\right]\\
&\le
E\left[\max_{m\in N} \log (m) \bigg(\sum_{i=1}^{m}
  g(\theta_{(i)})\bigg)\right]\\
&\leq P\Bigl(|m(\theta)-\tfrac{1}{2}n|>\epsilon n\Bigr)n\log (n)\\
&\quad+E\left[\max_{m:|m- \tfrac{1}{2}n|\leq \epsilon n}  \log (m)\bigg(\sum_{i=1}^{m}
  g(\theta_{(i)})\bigg)\right]
\end{align*}

By Lemma \ref{lemmalog}(ii), the first part is $\sim O(\log(n))$. Now we compute the second fart as follows,
\begin{align*}
 & E\left[\max_{m:|m- \tfrac{1}{2}n|\leq \epsilon n}  \log m\bigg(\sum_{i=1}^{m}
  g(\theta_{(i)})\bigg)\right]\\
&\le  \log \Big((\tfrac{1}{2}+\epsilon)n\Big)\cdot\\
  & \quad E\left[\max_{m:|m- \tfrac{1}{2}n|\leq \epsilon n}  \bigg(\sum_{i=1}^{(\tfrac{1}{2}-\epsilon)n}
  g(\theta_{(i)})+\sum_{i=1+(\tfrac{1}{2}-\epsilon)n}^{m}
  g(\theta_{(i)})\bigg)\right]\\
&\le\log  \Big((\tfrac{1}{2}+\epsilon)n\Big)E\left[\sum_{i=1}^{(\tfrac{1}{2}-\epsilon)n}
  g(\theta_{(i)})+2\epsilon n\right].
\end{align*}
Using Lemma \ref{orderstat}, we see that the second part  is  $\sim \frac{n}{4}\log(\frac{n}{2})$. Therefore, the profit achieved by the differentiated pricing scheme is $\sim\frac{n}{4}\log(\frac{n}{2})$.
\end{proof}

\section{Proof of Corollary \ref{asymptoticlinear}}\label{asymptoticlinearproof}
Similar to \eqref{opt:com:log}, the platform's optimization problem under complete information is
\begin{equation}\label{opt:com:linear}
  \max_{m\in N}\ \bigg(m\sum_{i=1}^m \theta_i - mc\bigg).
\end{equation}
Similar to \eqref{eq:munbouned}, the platform's optimization problem  under incomplete information and differentiated pricing scheme is
\begin{equation}
\int \max_{m\in N} \bigg( m \sum_{i=1}^{m}
  g(\theta_{(i)}) -mc\bigg) \dif F^n(\theta).\label{eq:munbounedlinear}
\end{equation}
Similar to \eqref{opt:incom:uni2}, the platform's optimization problem  under incomplete information and uniform pricing scheme is
\begin{equation}\label{opt:incom:unilinear}
\max_{\bar\theta}\bar{\theta}(1-\bar{\theta})^2n^2-n(1-\bar{\theta})c.
\end{equation}
As $n\rightarrow\infty$, \eqref{opt:incom:unilinear} becomes,
\[
\max_{\bar\theta}\bar{\theta}(1-\bar{\theta})^2n^2.
\]
The optimum is attained at $\bar{\theta}^*\rightarrow\tfrac{1}{3}$ and the optimal profit is $\sim 4/27n^2$. This proves the first part of Corollary \ref{asymptoticlinear}. For the rest part, we introduce some Lemmas.

\begin{lemma}\label{lemma}\ \\
(i) Let
\begin{align*}
  m(\theta) &= \arg\max_{m\in N} m\left(\sum_{i=1}^m
    g(\theta_{(i)})-c\right),\\[6pt]
  h(m) &=  \left[m\left(\sum_{i=1}^m
    g(\theta_{(i)})-c\right)\right]\\
    &\quad-\left[ (m-1)\left(\sum_{i=1}^{m-1}
    g(\theta_{(i)})-c\right)\right]\\
&= m  g(\theta_{(m)})+\sum_{i=1}^{m-1}
    g(\theta_{(i)}) - c.
\end{align*}
Then for $m\geq1$,
\begin{align*}
 m(\theta)\geq m&\implies h(m)\geq 0 \mbox{ or }g(\theta_{(m)})\ge0,\\
m(\theta)< m &\implies h(m)<0 .
\end{align*}
(ii) For $m=a n$, $0<a<1$,
\[
Eh(m) \sim (2-3a)a n.
\]
(iii) Suppose $m= (\tfrac{2}{3}-\epsilon)n$, where $\epsilon$ is
 a small positive number. Then
\begin{align*}
  P( h(m)\leq 0)&\leq \frac{\var[h(m)]}{(3\epsilon)^2 (\tfrac{2}{3}-\epsilon)^2 n^2}.
\end{align*}
Suppose $m= (\tfrac{2}{3}+\epsilon)n$. Then
\begin{align*}
  P( h(m)\geq0)&\leq \frac{\var[h(m)]}{(3\epsilon)^2 (\tfrac{2}{3}+\epsilon)^2 n^2}.
\end{align*}
(iv) For $m=a n$, $0<a<1$, \[
\var[h(m)] = O(n).
\]
(v) For $m=a n$, $a\ge2/3+\epsilon$, \[
  P( g(m)\geq0) = O(1/n).
\]
(vi)
\[
P\Bigl(|m(\theta)-\tfrac{2}{3}n|>\epsilon n\Bigr) = O(1/n).
\]
\end{lemma}
\begin{proof}
  The truth of (i) is straightforward. Note that $m(\theta)$ is the greatest $m$ such that $h(m)\geq0$ and note that $h(m)$ increases in $m$ when $g(\theta_{(m)})$ is nonnegative and decreases in $m$ when $g(\theta_{(m)})$ is negative.

   For (ii), recall that
  $g(\theta_i)=2\theta_i-1$. Thus, $E[g(\theta_{(i)})]=2(1-i/(n+1))-1$ and hence
\begin{align*}
E\left[m g(\theta_{(m)})+\sum_{i=1}^{m-1}g(\theta_{(i)})\right]
&=2m\left(1-\frac{m}{n+1}\right)-m\\
&\quad +\sum_{i=1}^{m-1}2\left(1-\frac{i}{n+1}\right)-(m-1)\\
&= -\frac{3 m^2-2 m n-3 m+n+1}{n+1}\\
\intertext{and so if $m= an$}
E[h(m)]&\sim (2-3a)an.
\end{align*}
%and so negative if $a>2/3$.

\noindent For (iii) we use Chebyshev's
inequality.
Suppose $m\leq (\tfrac{2}{3}-\epsilon)n$. Then
\begin{align*}
  P( h(m)\leq0)&= P( h(m)-E[h(m)]\leq-E[h(m)])\\[4pt]
&\leq  P( |h(m)-E[h(m)]|\geq E[h(m)])\\[4pt]
&\leq \frac{\var[h(m)]}{(3\epsilon)^2 (\tfrac{2}{3}-\epsilon)^2n^2}.
\end{align*}
Suppose $m\geq (\tfrac{2}{3}+\epsilon)n$. Then
\begin{align*}
  P( h(m)\geq0)&= P( h(m)-E[h(m)]\geq-E[h(m)])\\[4pt]
&\leq  P( |h(m)-E[h(m)]|\geq -E[h(m)])\\[4pt]
&\leq \frac{\var[h(m)]}{(3\epsilon)^2 (\tfrac{2}{3}+\epsilon)^2 n^2}.
\end{align*}

\noindent For (iv) we find $\var[h(m)]$.
This is
\begin{align*}
  \var[h(m)] & = m^2 \var[ g(\theta_{(m)})]+\sum_{i=1}^{m-1}
    \var[g(\theta_{(i)})]\\
&\quad +2m \sum_{i=1}^{m-1}\cov[ g(\theta_{(i)}), g(\theta_{(m)})]\\
&\quad +2\sum_{1\leq i<j\leq m-1} \cov[ g(\theta_{(i)}), g(\theta_{(j)})].
\end{align*}
An evaluation of this for $m=an$ gives
\[
\var[h(m)]  \sim \left(\frac{28 a^3}{3}-9
   a^4\right)n.
\]
The term in parentheses is positive.

For (v), recall that
  $g(\theta_m)=2\theta_m-1$.
  If $m\ge(\tfrac{2}{3}+\epsilon)n$,
  then $E[g(\theta_{(m)})]=2(1-m/(n+1))-1<0$ and $\var[g(\theta_{(m)})]=4\frac{m(n+1-m)}{(n+1)^2(n+2)}$. It follows that
\begin{align*}
  P( g(\theta_{(m)})>0)&= P( g(\theta_{(m)})-E[g(\theta_{(m)})]>-E[g(\theta_{(m)})])\\
&\leq  P( |g(\theta_{(m)})-E[g(\theta_{(m)})]|\geq -E[g(\theta_{(m)})])\\
&\leq\frac{\var[g(\theta_{(m)})]}{E^2[g(\theta_{(m)})]}=\frac{4\frac{m(n+1-m)}{(n+1)^2(n+2)}}{(1-2m/(n+1))^2}.
\end{align*}
Thus, when $m=an\ge(\tfrac{2}{3}+\epsilon)n$, we have $  P( g(\theta_{(m)})>0)\sim O(1/n)$.

For (vi), note that
\begin{align*}
&P\Bigl(|m(\theta)-\tfrac{2}{3}n|>\epsilon n\Bigr)\\
&=P\Bigl(m(\theta)>(\tfrac{2}{3}+\epsilon) n\Bigr)+P\Bigl(m(\theta)<(\tfrac{2}{3}-\epsilon) n\Bigr)\\
&\leq P\Bigl(h((\tfrac{2}{3}+\epsilon) n)\geq0\mbox{ or }g(\theta_{(m)})\geq 0\Bigr)+P\Bigl(h((\tfrac{2}{3}-\epsilon) n)<0\Bigr)\\
&\leq P\Bigl(h((\tfrac{2}{3}+\epsilon) n)\geq0)+P(g(\theta_{(m)})\geq 0)\\
&\quad+P\Bigl(h((\tfrac{2}{3}-\epsilon) n)<0\Bigr).
\end{align*}
From this, (vi) follows from
application of (iii), (iv) and (v).
\end{proof}
Now we prove the theorem.
\begin{proof}
We use below that $g\leq 1$. Now
\begin{align*}
 \int& \max_{m\in N} m \bigg(\sum_{i=1}^{m}
  g(\theta_{(i)}) -c\bigg) d F^n(\theta)\\
&=
E\left[\max_{m\in N} m \bigg(\sum_{i=1}^{m}
  g(\theta_{(i)})-c\bigg)\right]\\
&\leq
E\left[\max_{m\in N} m \bigg(\sum_{i=1}^{m}
  g(\theta_{(i)})\bigg)\right]\\
&\leq P\Bigl(|m(\theta)-\tfrac{2}{3}n|>\epsilon n\Bigr)n^2\\
&\quad+ E\left[\max_{m:|m- \tfrac{2}{3}n|\leq \epsilon n}  m\bigg(\sum_{i=1}^{m}
  g(\theta_{(i)})\bigg)\right]\\
&\leq P\Bigl(|m(\theta)-\tfrac{2}{3}n|>\epsilon
n\Bigr)n^2\\
&\quad+(\tfrac{2}{3}+\epsilon)nE\left[ 2\epsilon n+\sum_{i=1}^{(\tfrac{2}{3}-\epsilon)n}
  g(\theta_{(i)})\right].
\end{align*}
Using Lemma \ref{lemma} (v), and the fact that $\epsilon$ is arbitrary, we see
that the right hand side $\sim (4/27)n^2$.
\end{proof}
\section{Proof of Proposition \ref{poiu}}\label{poiuproof}
\begin{proof}
Assume user' utility is given by \eqref{upgeneral}.
  As $n\rightarrow\infty$,
\begin{align}
  E_\theta (\Pi) &=   E_\theta \Bigg(\max_{m\in N}\ v(m)\sum_{i=1}^m \theta_i - mc\Bigg)\nonumber\\
  &= \int \bigg(\max_{m\in N}\ v(m)\sum_{i=1}^m \theta_i - mc \bigg) d F^n(\theta)\notag\\
  &\sim v(n)\int\sum_{i=1}^{n}\theta_id F^n(\theta)-nc\notag\\
  &\sim\tfrac{1}{2}nv(n).\nonumber
\end{align}

Now we compute the maximum profit under incomplete information and uniform price,
\[
\Pi_U=\max_{\bar{\theta}\in[0,1]}\bar{\theta}(1-\bar{\theta})nv(n(1-\bar{\theta}))-\bar{\theta}v(n(1-\bar{\theta}))c.
\]
We can omit the second part as $n\rightarrow\infty$,
\[
\lim\limits_{n\rightarrow\infty}\Pi_U=\max_{\bar{\theta}\in[0,1]}\bar{\theta}(1-\bar{\theta})nv(n(1-\bar{\theta})).
\]

The price of information is
$$\PoI=\lim_{n\rightarrow\infty}\frac{\tfrac{1}{2}v(n)}{\max\limits_{\bar{\theta}\in[0,1]}\bar{\theta}(1-\bar{\theta})v(n(1-\bar{\theta}))}.$$
Note that $v(x)$ is concave and hence $V(n(1-\bar{\theta}))\ge(1-\bar{\theta})v(n)$. It follows that
\begin{align}
&\max\limits_{\bar{\theta}\in[0,1]}\bar{\theta}(1-\bar{\theta})v(n(1-\bar{\theta}))\notag\\
&\geq\max\limits_{\bar{\theta}\in[0,1]}\bar{\theta}(1-\bar{\theta})^2v(n)=\tfrac{4}{27}v(n).\label{linearpoi}
\end{align}
Therefore,
$$\PoI=\lim_{n\rightarrow\infty}\frac{\tfrac{1}{2}v(n)}{\max\limits_{\bar{\theta}\in[0,1]}\bar{\theta}(1-\bar{\theta})V(n(1-\bar{\theta}))}\le\frac{\tfrac{1}{2}v(n)}{\tfrac{4}{27}v(n)}=\frac{27}{8}.$$

Note that
\begin{equation}\label{logpoi}
\max\limits_{\bar{\theta}\in[0,1]}\bar{\theta}(1-\bar{\theta})v(n(1-\bar{\theta}))\le\max\limits_{\bar{\theta}\in[0,1]}\bar{\theta}(1-\bar{\theta})V(n)=\frac{1}{4}v(n).
\end{equation}
Therefore,
\[\PoI=\lim_{n\rightarrow\infty}\frac{\tfrac{1}{2}v(n)}{\max\limits_{\bar{\theta}\in[0,1]}\bar{\theta}(1-\bar{\theta})v(n(1-\bar{\theta}))}\ge\frac{\tfrac{1}{2}v(n)}{\frac{1}{4}v(n)}=2.\]
Note that  equality in \eqref{linearpoi} holds for linear function, thus $\PoI$ for linear utility model \eqref{upu} is $27/8$.  Equality in \eqref{logpoi} holds for logarithm function as $n\rightarrow\infty$, that is, for any $\bar{\theta}\in[0,1]$
\[
\lim_{n\rightarrow\infty}\frac{v(n(1-\bar{\theta}))}{v(n)}=1.
\]
Thus, $\PoI$ for  logarithm utility model \eqref{uplog} is $2$.
\end{proof}

\section{Proof of Propositon \ref{two:cominfo:opt}}\label{two:cominfo:optproof}
\eqref{opt:CompInfo_two} can be written as
\begin{equation}\label{opt:CompInfo_twocopy}
  \max_{m_1,m_2}\ m_1\left(\bigg(\sum_{i=1}^{m_1} \theta_i+\sum_{i=n_1+1}^{n_1+m_2} \theta_i\bigg)-\frac{m_1+m_2}{m_1}c\right).
\end{equation}
 We first maximize the term in the bracket of \eqref{opt:CompInfo_twocopy}. Note that
\begin{align*}
 B(m_1,m_2)=&\bigg(\sum_{i=1}^{m_1} \theta_i+\sum_{i=n_1+1}^{n_1+m_2} \ \theta_i\bigg)-\frac{m_1+m_2}{m_1}c\\
 \le &\bigg(\sum_{i=1}^{n_1} \theta_i+\sum_{i=n_1+1}^{n_1+m_2} \theta_i\bigg)-\frac{n_1+m_2}{n_1}c.
 \end{align*}
We only need to maximize $B(n_1,m_2)$ over $m_2$.   $B(n_1,m_2)$ is maximized for $m_2=\bar m_2$ where $\bar m_2$ is the largest user number $m_2$ such that $\theta_{n_1+m_2}\ge c/n_1$. Thus, if $B(n_1,\bar m_2)>0$ then the maximal profit is $n_1B(n_1,\bar m_2)$ which is achieved for $m_1=n_1$ and $m_2=\bar m_2$. Otherwise, the maximal profit is 0 and it is optimal to include no user.

\section{Proof of Theorem \ref{two:UniPri}}\label{two:UniPriproof}
We use $k$ to denote $n_1/n_2$ throoughout the proof. It is straightforward to check that the thresholds $\bar\theta_1^*=\min\{\tfrac{2}{3}-\frac{\sqrt{k(4k+3)}}{6k},0\}$ and $\bar\theta_2^*=\tfrac{1}{2}$ solve the following optimization problem,
\begin{align*}
\max_{\bar\theta_1,\bar\theta_2\in[0,1]}&kn_2(1-\bar\theta_1)(\bar\theta_1(1-\bar\theta_1)kn_2)n_2)\\
&+n_2(1-\bar\theta_2)(\bar\theta_2(1-\bar\theta_1)kn_2).
\end{align*}
Direct calculation will prove the first part of the theorem.

Now we prove that $\Pi_D\sim\eqref{maxprofit}$. Since $\Pi_D\ge\Pi_U$, we only need to prove that $\Pi_D$ is bounded above by \eqref{maxprofit}.
Define $\eta_1=1-\bar\theta_1^*=\min\{\tfrac{1}{3}+\frac{\sqrt{k(4k+3)}}{6k},1\}$ and $\eta_2=1-\bar\theta_2^*=\tfrac{1}{2}$.
We first prove some lemmas.
\begin{lemma}\label{two:lemma}\ \\
(i) Let $m_1=a_1n_1$, $m_2=a_2n_2$, and
\begin{align*}
 h_1(m_1,m_2)&=\sum_{i=1}^{m_1}g(\theta_{(i)})+\sum_{i=1}^{m_2}g(\theta_{(n_1+i)})-c+m_1g(\theta_{(m_1)}),\\
h_2(m_1,m_2)&=m_1g(\theta_{(n_1+m_2)})-c.
\end{align*}
  Then, for any $a_1\in(0,1)$  and $a_2\in(0,1)$,
  \begin{align*}
    E[h_1(m_1,m_2)]&\sim n_2(-3ka_1^2-a_2^2+2ka_1+a_2),\\
  E[h_2(m_1,m_2)]&\sim n_2(-2ka_1a_2+ka_1).
  \end{align*}
  and
  \[\var[h_1(m_1,m_2)]=O(1/n_2), \qquad\var[h_2(m_1,m_2)]=O(1/n_2).\]
  (ii)    Let
    \begin{align*}
        f_1(a_2)&=\min\{\frac{k+\sqrt{k^2+3ka_2-3ka_2^2}}{3k},1\},
        \end{align*}
        Then, as $n_2\rightarrow\infty$,
        \begin{equation*}
        \begin{cases}
        0<a_1<f_1(a_2)\implies  E[h_1(m_1,m_2)]>0\\
        f_1(a_2)<a_1\le 1\implies  E[h_1(m_1,m_2)]<0.
        \end{cases}
        \end{equation*}
        and
         \begin{equation*}
        \begin{cases}
        0< a_2<\tfrac{1}{2}\implies  E[h_2(m_1,m_2)]>0\\
        \tfrac{1}{2}<a_2\le 1\implies  E[h_2(m_1,m_2)]<0.
        \end{cases}
        \end{equation*}

(iii) Let $R(\epsilon_1, \epsilon_2)$ be a region in $[0,1]^2$. $R(\epsilon_1, \epsilon_2)$ is defined by the following system of inequalities,
       \begin{equation}\label{region}
\begin{cases}
f_1(a_2)-\epsilon_1\le a_1 \le f_1(a_2)+\epsilon_1,\\
\tfrac{1}{2}-\epsilon_2\le a_2 \le \tfrac{1}{2}+\epsilon_2.
\end{cases}
         \end{equation}
Then, for arbitrary $\epsilon>0$, there are some  $\epsilon_1,\epsilon_2>0$ such that
 \[R(\epsilon_1, \epsilon_2)\subset[\eta_1-\epsilon,\eta_1+\epsilon]\times[\eta_2-\epsilon,\eta_2+\epsilon].\]
\end{lemma}
\begin{proof}
 (i) can be proved by applying the results of Lemma \ref{orderstat} directly, we omit the detailed calculation. Straightforward calculation  shows (ii) is true. (iii) follows from the facts that $f_1(a_2)=\eta_1$  when $a_2=\tfrac{1}{2}$ and that $f_1(a_2)$ is continuous.
\end{proof}

Now we prove the Theorem.
\begin{proof}
Let
 \begin{align*}
&(m_1(\theta),(m_2(\theta)) \\
&=\arg\max_{(m_1,m_2)} \Bigg[ m_1\bigg(\sum_{i=1}^{m_1}
  g(\theta_{(i)}) +\sum_{i=1}^{m_2}
  g(\theta_{(n_1+i)})\bigg)\\
&\quad-(m_1+m_2) c \Bigg],
  \end{align*}
  and $(a_1(\theta),a_2(\theta))=(m_1(\theta)/n_1,(m_2(\theta)/n_2)$.
Note that $g\le1$. Then,
 \begin{align*}
 \int &\max_{m_1,m_2}  m_1\bigg(\sum_{i=1}^{m_1}
  g(\theta_{(n_1+i)}) +\sum_{i=1}^{m_2}
  g(\theta_{2(i)})\bigg)\\
  & -(m_1+m_2)c\bigg)d F^n(\theta)\\
&  \leq E\Bigg[  \max_{m_1,m_2}m_1\bigg(\sum_{i=1}^{m_1}
  g(\theta_{1(i)}) +\sum_{i=1}^{m_2}
  g(\theta_{2(i)})\bigg)\Bigg] \\
   &\leq P\Bigg(|m_1(\theta)-\eta_1 n_1|>\epsilon n_1\mbox{ or }|m_2(\theta)-\eta_2 n_2|>\epsilon n_2\Bigg)\cdot\\
  &\quad k(k+1)n_2^2+\\
  &\quad E\Bigg[\max_{m_1:|m_1-\eta_1 n_1|\leq\epsilon n_1\atop m_2:|m_1-\eta_2 n_2|\leq\epsilon n_2}m_1\bigg(\sum_{i=1}^{m_1}
  g(\theta_{(i)}) +\sum_{i=n_1+1}^{n_1+m_2}
  g(\theta_{(i)})\bigg)\Bigg ].
 \end{align*}
 Note that the second part
 \begin{align*}
 &E\Bigg[\max_{m_1:|m_1-\eta_1 n_1|\leq\epsilon n_1\atop m_2:|m_1-\eta_2 n_2|\leq\epsilon n_2}m_1\bigg(\sum_{i=1}^{m_1}
  g(\theta_{1(i)}) +\sum_{i=1}^{m_2}
  g(\theta_{2(i)})\bigg)\Bigg] \\
   &\leq (\eta_1+\epsilon)n_1\cdot\\
  &\quad E\Bigg[2\epsilon n_1+\sum_{i=1}^{(\eta_1-\epsilon)n_1}
  g(\theta_{1(i)}) +2\epsilon n_2+\sum_{i=1}^{(\eta_2-\epsilon)n_2}
  g(\theta_{2(i)}))\Bigg]\\
  &\sim k\eta_1(k(1-\eta_1)\eta_1+(1-\eta_2)\eta_2)n_2^2
 \end{align*}
The last step is derived by Lemma \ref{orderstat} and the fact that $\epsilon$ is arbitrary small.
If we can further show that for arbitrary $\epsilon>0$,
 $$P\Bigg(|m_1(\theta)-\eta_1 n_1|>\epsilon n_1\mbox{ or }|m_2(\theta)-\eta_2 n_2|>\epsilon n_2\Bigg)\sim O(1/n_2),$$
or
 \begin{align*}
 &P\Bigg(|m_1(\theta)-\eta_1 n_1|\le\epsilon n_1\mbox{ and }|m_2(\theta)-\eta_2 n_2|\le\epsilon n_2\Bigg)\\
& \sim1- O(1/n_2),
 \end{align*}
then the first part is $\sim O(n_2)$ which is dominated by the second part and hence the theorem is true.

To prove the statement above, we will first show that for arbitrary $\epsilon_1,\epsilon_2>0$, $(a_1(\theta),a_2(\theta))$ is in the region $R(\epsilon_1,\epsilon_2)$ defined in Lemma \ref{two:lemma}(iii) with probability $\sim1- O(1/n_2)$.

For any $a_2\in[0,1]$, from Lemma \ref{two:lemma}(i), (ii), and Chebyshev's inequality it follows that
 \begin{align*}
 &P\bigg(a_1(\theta)\in[f_1(a_2(\theta))-\epsilon_1,f_1(a_2(\theta))+\epsilon_1]\bigg|a_2(\theta)=a_2\bigg)\\
& \sim1-O(\frac{1}{n_2})
 \end{align*}
 The detailed proof is similar to the proof of Lemma 6, we do not repeat it here. By integral, we have
  \begin{align*}
  P\bigg(a_1(\theta)\in[f_1(a_2(\theta))-\epsilon_1,f_1(a_2(\theta))+\epsilon_1]\bigg)
  \sim1-O(\frac{1}{n_2}).
   \end{align*}
Similarly,  for any $a_1\in[0,1]$, from Lemma \ref{two:lemma}(i), (ii), and Chebyshev's inequality it follows that
 \begin{align*}
 &P\bigg(a_2(\theta)\in[\tfrac{1}{2}-\epsilon_2,\tfrac{1}{2}+\epsilon_2]\bigg|a_1(\theta)=a_1\bigg)\sim1-O(\frac{1}{n_2}).
  \end{align*}
By integral, we have
  $$P\bigg(a_2(\theta)\in[\tfrac{1}{2}-\epsilon_2,\tfrac{1}{2}+\epsilon_2]\bigg)\sim1-O(\frac{1}{n_2}).$$
Thus,
\[P\bigg((a_1(\theta),a_2(\theta))\in R(\epsilon_1,\epsilon_2)\bigg)\sim1-O(\frac{1}{n_2}).\]

It follows from Lemma \ref{two:lemma}(iii) that for arbitrary $\epsilon>0$, there are some $\epsilon_1,\epsilon_2>0$ such that
      \[R(\epsilon_1, \epsilon_2)\subset[\eta_1-\epsilon,\eta_1+\epsilon]\times[\eta_2-\epsilon,\eta_2+\epsilon].\]
      Therefore,
         \begin{align*}
         &P\Bigg(|m_1(\theta)-\eta_1 n_1|\le\epsilon n_1\mbox{ and }|m_2(\theta)-\eta_2 n_2|\le\epsilon n_2\Bigg)\\
         &\sim1- O(1/n_2).\qedhere
         \end{align*}
         \end{proof}
\section{Proof of Theorem \ref{two:poi}} \label{two:poiproof}
\begin{proof}
As $n\rightarrow\infty$,
\begin{align*}
E_{\theta}(\Pi) &= E\left[\max_{m_1,m_2}\ m_1\bigg(\sum_{i=1}^{m_1} \theta_i+\sum_{i=n_1+1}^{n_1+m_2} \theta_i\bigg)-(m_1+m_2)c\right]\\
&\sim E\left[\max_{m_1,m_2}\ m_1\bigg(\sum_{i=1}^{m_1} \theta_i+\sum_{i=n_1+1}^{n_1+m_2} \theta_i\bigg)\right]\\
&=k(k+1)n_2^2.
\end{align*}
Since $E_U$ is given by \eqref{maxprofit},  the price of information is
\[\PoI=\frac{54k(k+1)}{(2k+\sqrt{k(4k+3)})(3+2k+\sqrt{k(4k+3)})},\]
when $k\ge\tfrac{1}{4}$,
and
\[\PoI={2(k+1)},\]
when $0<k<\tfrac{1}{4}$.
The derivative of $\PoI$ with respect to $k$ is less than zero, hence it decreases with $k$. Replacing $k$ by ${n_1}/{n_2}$ will prove the theorem.
\end{proof}

\end{document}